\documentclass{sig-alternate-05-2015}

\usepackage{balance}
\usepackage{algorithm,algorithmic}
\usepackage{graphicx}
\usepackage{multirow}
\usepackage{appendix}
\usepackage{array}
\usepackage{ctable}
\usepackage[size=small,format=plain,labelfont=up,textfont=up]{caption}
\DeclareCaptionType{copyrightbox}
\usepackage{caption}
\usepackage{subcaption}
\usepackage{filecontents}
\usepackage{amsmath,amssymb}
\usepackage{hyperref}
\usepackage{mathtools}

\graphicspath{{./figs/}}

\long\def\BL{\textsc{Baseline}}
\long\def\LM{\textsc{Limit}}
\long\def\LMS{\textsc{Limit-SS}}
\long\def\LMC{\textsc{Limit-CB}}

\newcommand{\tn}{\tabularnewline}
\newcommand{\INDSTATE}[1][1]{\STATE\hspace{#1\algorithmicindent}}

\newcommand\fnurl[2]{%
\href{#2}{#1}\footnote{\url{#2}}%
}

\newfont{\mycrnotice}{ptmr8t at 7pt}
\newfont{\myconfname}{ptmri8t at 7pt}
\let\crnotice\mycrnotice%
\let\confname\myconfname%

\title{Differentially Private Publication of Location Entropy}

\def\sharedaffiliation{%
\end{tabular}
\begin{tabular}{c}}
\numberofauthors{3}
    \author{
      \alignauthor Hien To \thanks{These authors contributed equally to this work.}\\      
      \email\normalsize{hto@usc.edu}
      \alignauthor Kien Nguyen \footnotemark[1] \\
      \email\normalsize{kien.nguyen@usc.edu}
      \alignauthor Cyrus Shahabi \\
      \email\normalsize{shahabi@usc.edu}
      \sharedaffiliation
       \affaddr{Department of Computer Science, University of Southern California}\\
       \affaddr{Los Angeles, CA 90089}\\
          }

\begin{document}

\CopyrightYear{2016}
\setcopyright{acmcopyright}
\conferenceinfo{SIGSPATIAL'16,}{October 31-November 03, 2016,
Burlingame, CA, USA}
\isbn{978-1-4503-4589-7/16/10}\acmPrice{\$15.00}
\doi{http://dx.doi.org/10.1145/2996913.2996985}

\maketitle
\begin{abstract}

Location entropy (LE) is a popular metric for measuring the popularity of various locations (e.g., points-of-interest). Unlike other metrics computed from only the number of (unique) visits to a location, namely {\em frequency}, LE also captures the {\em diversity} of the users' visits, and is thus more accurate than other metrics. Current solutions for computing LE require full access to the past visits of users to locations, which poses privacy threats. This paper discusses, for the first time, the problem of perturbing location entropy for a set of locations according to differential privacy. The problem is challenging because removing a single user from the dataset will impact multiple records of the database; i.e., all the visits made by that user to various locations. Towards this end, we first derive non-trivial, tight bounds for both local and global sensitivity of LE, and show that to satisfy $\epsilon$-differential privacy, a large amount of noise must be introduced, rendering the published results useless.  Hence, we propose a thresholding technique to limit the number of users' visits, which significantly reduces the perturbation error but introduces an approximation error. To achieve better utility, we extend the technique by adopting two weaker notions of privacy: smooth sensitivity (slightly weaker) and crowd-blending (strictly weaker). Extensive experiments on synthetic and real-world datasets show that our proposed techniques preserve original data distribution without compromising location privacy.

\end{abstract}
\category{H.2.4}{Database Management}{Database Applications}[Spatial databases and GIS]
\category{H.1.1}{Models and Principles}{Systems and Information Theory}[Information theory]
\keywords{Differential Privacy, Location Entropy}

\section{Introduction}

Due to the pervasiveness of GPS-enabled mobile devices and the popularity of location-based services such as mapping and navigation apps (e.g., Google Maps, Waze), or spatial crowdsourcing apps (e.g., Uber, TaskRabbit), or apps with geo-tagging (e.g., Twitter, Picasa, Instagram, Flickr), or check-in functionality (e.g., Foursquare, Facebook), numerous industries are now collecting fine-grained location data from their users. While the collected location data can be used for many commercial purposes by these industries (e.g., geo-marketing), other companies and non-profit organizations (e.g., academia, CDC) can also be empowered if they can use the location data for the greater good (e.g., research, preventing the spread of disease). Unfortunately, despite the usefulness of the data, industries do not publish their location data due to the sensitivity of their users' location information. However, many of these organizations do not need access to the raw location data but aggregate or processed location data would satisfy their need.

One example of using location data is to measure the popularity of a location that can be used in many application domains such as public health, criminology, urban planning, policy, and social studies. One accepted metric to measure the popularity of a location is location entropy (or LE for short).  LE captures both the frequency of visits (how many times each user visited a location) as well as the diversity of visits (how many unique users visited a location) without looking at the functionality of that location; e.g., is it a private home or a coffee shop? Hence, LE has shown that it is able to better quantify the popularity of a location as compared to the number of unique visits or the number of check-ins to the location~\cite{cranshaw2010bridging}. For example, ~\cite{cranshaw2010bridging} shows that LE is more successful in accurately predicting friendship from location trails over simpler models based only on the number of visits.
LE is also used to improve online task assignment in spatial crowdsourcing~\cite{kazemi2012geocrowd,to2015server} by giving priority to workers situated in less popular locations because there may be no available worker visiting those locations in the future.

Obviously, LE can be computed from raw location data collected by various industries; however, the raw data cannot be published due to serious location privacy implications~\cite{ghinita2008private,DeMontjoye2013locationunique,to2014framework}. Without privacy protection, a malicious adversary can stage a broad spectrum of attacks such as physical surveillance and stalking, and breach of sensitive information such as an individual's health issues (e.g., presence in a cancer treatment center), alternative lifestyles, political and religious preferences (e.g., presence in a church). Hence, in this paper we propose an approach based on differential privacy (DP)~\cite{dwork2006differential} to publish LE for a set of locations without compromising users' raw location data. DP has emerged as the de facto standard with strong protection guarantees for publishing aggregate data. It has been adapted by major industries for various tasks without compromising individual privacy, e.g., data analytics with Microsoft [15], discovering users' usage patterns with \fnurl{Apple}{https://www.wired.com/2016/06/apples-differential-privacy-collecting-data/}, or crowdsourcing statistics from end-user client software~\cite{erlingsson2014rappor}
and training of deep neural networks~\cite{abadi2016deep} 
with Google.
DP ensures that an adversary is not able to reliably learn from the published sanitized data whether or not a particular individual is present in the original data, regardless of the adversary's prior knowledge.

It is sufficient to achieve $\epsilon$-DP ($\epsilon$ is privacy loss) by adding Laplace \textit{noise} with mean zero and scale proportional to the \textit{sensitivity} of the query (LE in this study)~\cite{dwork2006differential}. The sensitivity of LE is intuitively the maximum amount that one individual can impact the value of LE. The higher the sensitivity, the more noise must be injected to guarantee $\epsilon$-DP.
Even though DP has been used before to compute Shannon Entropy~\cite{Blum:2005:SuLQ} (the formulation adapted in LE), the main challenge in differentially private publication of LE is that adding (or dropping) a single user from the dataset would impact multiple entries of the database, resulting in a high sensitivity of LE. To illustrate, consider a user that has contributed many visits to a single location; thus, adding or removing this user would significantly change the value of LE for that location. Alternatively, a user may contribute visits to multiple locations and hence impact the entropy of all those visited locations.
Another unique challenge in publishing LE (vs. simply computing the Shannon Entropy) is due to the presence of skewness and sparseness in real-world location datasets where the majority of locations have small numbers of visits.

Towards this end, we first compute a non-trivial tight bound for the global sensitivity of LE. Given the bound, a sufficient amount of noise is introduced to guarantee $\epsilon$-DP.
However, the injected noise linearly increases with the maximum number of locations visited by a user (denoted by $M$) and monotonically increases with the maximum number of visits a user contributes to a location (denoted by $C$), and such an excessive amount of noise renders the published results useless.
We refer to this algorithm as {\BL}.
Accordingly, we propose a technique, termed {\LM}, to limit user activity by thresholding $M$ and $C$, which significantly reduces the perturbation error. Nevertheless, limiting an individual's activity entails an approximation error in calculating LE. These two conflicting factors require the derivation of appropriate values for $M$ and $C$ to obtain satisfactory results. We empirically find such optimal values.

Furthermore, to achieve a better utility, we extend {\LM} by adopting two weaker notions of privacy: smooth sensitivity~\cite{Nissim:2007:SmoothSensitivity} and crowd-blending~\cite{gehrke2012crowd} (strictly weaker). We denote the techniques as {\LMS} and {\LMC}, respectively.
{\LMS} provides a slightly weaker privacy guarantee, i.e., $(\epsilon,\delta)$-differential privacy by using local sensitivity with much smaller noise magnitude. %
We propose an efficient algorithm to compute the local sensitivity of a particular location that depends on $C$ and the number of users visiting the location (represented by $n$) such that the local sensitivity of all locations can be precomputed, regardless of the dataset.
Thus far, we publish entropy for all locations; however, the ratio of noise to the true value of LE (noise-to-true-entropy ratio) is often excessively high when the number of users visiting a location $n$ is small (i.e., the entropy of a location is bounded by $\log(n)$). For example, given a location visited by only two users with an equal number of visits (LE is $\log 2$), removing one user from the database drops the entropy of the location to zero. To further reduce the noise-to-true-entropy ratio, {\LMC} aims to publish the entropy of locations with at least $k$ users ($n\ge k$) and suppress the other locations. By thresholding $n$, the global sensitivity of LE significantly drops, implying much less noise. We prove that {\LMC} satisfies ($k,\epsilon$)-crowd-blending privacy.

We conduct an extensive set of experiments on both synthetic and real-world datasets. We first show that the truncation technique ({\LM}) reduces the global sensitivity of LE by two orders of magnitude, thus greatly enhancing the utility of the perturbed results. We also demonstrate that {\LM} preserves the original data distribution after adding noise.
Thereafter, we show the superiority of {\LMS} and {\LMC} over {\LM} in terms of achieving higher utility (measured by KL-divergence and mean squared error metrics). Particularly, {\LMC} performs best on sparse datasets while {\LMS} is recommended over {\LMC} on dense datasets. We also provide insights on the effects of various parameters: $\epsilon,C,M,k$ on the effectiveness and utility of our proposed algorithms. Based on the insights, we provide a set of guidelines for choosing appropriate algorithms and parameters.

The remainder of this paper is organized as follows. In Section~\ref{sec:problem}, we define the problem of publishing LE according to differential privacy. Section~\ref{sec:prelim} presents the preliminaries. Section~\ref{sec:dple} introduces the baseline solution and our thresholding technique. Section~\ref{sec:relaxation} presents our utility enhancements by adopting weaker notions of privacy. Experimental results are presented in Section~\ref{sec:experiment}, followed by a survey of related work in Section~\ref{sec:related}, and conclusions in Section~\ref{sec:conclude}.

\section{Problem Definition}
\label{sec:problem}
In this section we present the notations and the formal definition of the problem.

Each location $l$ is represented by a point in two-dimensional space and a unique identifier $l$ $(-180 \le l_{lat} \le 180)$ and $(-90 \le l_{lon} \le 90)$\footnote{\small{$l_{lat},l_{lon}$ are real numbers with ten digits after the decimal point.}}. Hereafter, $l$ refers to both the location and its unique identifier.
For a given location $l$, let $O_l$ be the set of visits to that location. Thus, $c_l=|O_l|$ is the total number of visits to $l$. Also, let $U_l$ be the set of distinct users that visited $l$, and $O_{l,u}$ be the set of visits that user $u$ has made to the location $l$. Thus, $c_{l,u}=|O_{l,u}|$ denotes the number of visits of user $u$ to location $l$. The probability that a random draw from $O_l$ belongs to $O_{l,u}$ is $p_{l,u}=\frac{|c_{l,u}|}{|c_l|}$, which is the fraction of total visits to $l$ that belongs to user $u$. The location entropy for $l$ is computed from Shannon entropy \cite{shannon1948mathematical} as follows:
\begin{equation}
\label{eq:le}
H(l) = H(p_{l,u_1}, p_{l,u_2}, \dots, p_{l,u_{|U_l|}}) = - \sum_{u \in U_l} p_{l,u} \log p_{l,u}
\end{equation}
In our study the natural logarithm is used. A location has a higher entropy when the visits are distributed more evenly among visiting users, and vice versa.
Our goal is to publish location entropy of all locations $L=\{l_1, l_2, ..., l_{|L|}\}$, where each location is visited by a set of users $U=\{u_1, u_2, ..., u_{|U|}\}$, while preserving the location privacy of users.
Table~\ref{tab:notation} summarizes the notations used in this paper.

\begin{table}
\begin{center}
\footnotesize
\scriptsize{
\begin{tabular}{ l | l}
\hline
$l,L,|L|$ & a location, the set of all locations and its cardinality \tn
\hline
$H(l)$ & location entropy of location $l$ \tn
\hline
$\hat{H}(l)$ & noisy location entropy of location $l$ \tn
\hline
$\Delta H_l$ & sensitivity of location entropy for location $l$ \tn
\hline
$\Delta H$ & sensitivity of location entropy for all locations \tn
\hline
$O_l$ & the set of visits to location $l$ \tn
\hline
$u,U,|U|$ & a user, the set of all users and its cardinality \tn
\hline
$U_l$ & the set of distinct users who visits $l$ \tn
\hline
$O_{l,u}$ & the set of visits that user $u$ has made to location $l$ \tn
\hline
$c_l$  & the total number of visits to $l$ \tn
\hline
$c_{l,u}$  & the number of visits that user $u$ has made to location $l$  \tn
\hline
$C$ & maximum number of visits of a user to a location \tn
\hline
$M$ & maximum number of locations visited by a user \tn
\hline
$p_{l,u}$ & the fraction of total visits to $l$ that belongs to user $u$ \tn
\end{tabular}}
\caption{Summary of notations.}\label{tab:notation}
\vspace{-15pt}
\end{center}
\end{table}
\section{Preliminaries}
\label{sec:prelim}

We present Shannon entropy properties and the differential privacy notion that will be used throughout the paper.

\subsection{Shannon Entropy}

Shannon \cite{shannon1948mathematical} introduces entropy as a measure of the uncertainty in a random variable with a probability distribution $U = (p_1, p_2,...,p_{|U|})$:
\begin{equation}
\label{eq:entropy}
H(U) = -\displaystyle\sum_{i}p_{i} \log p_i
\end{equation}
where $\sum_{i} p_i=1$. $H(U)$ is maximal if all the outcomes are equally likely:
\begin{equation}
\label{eq:entropy_maximum}
H(U) \le H(\frac{1}{|U|}, ..., \frac{1}{|U|})  = \log |U|
\end{equation}
\textbf{Additivity Property of Entropy:} 
Let $U_1$ and $U_2$ be non-overlapping partitions of a database $U$ including users who contribute visits to a location $l$, and $\phi_1$ and $\phi_2$ are probabilities that a particular visit belongs to partition $U_1$ and $U_2$, respectively.
Shannon discovered that using logarithmic function preserves the  \emph{additivity} property of entropy: 
\begin{equation}
\label{eq:entropy_additivity}
H(U) = \phi_1H(U_1) + \phi_2H(U_2) + H(\phi_1, \phi_2)  \nonumber
\end{equation}
Subsequently, adding a new person $u$ into $U$ changes its entropy to: 
\begin{equation}
\label{eq:entropy_add}
\centering H(U^+) = \frac{c_l}{c_l + c_{l,u}}H(U) + H\Big(\frac{c_{l,u}}{c_l+c_{l,u}}, \frac{c_l}{c_l+c_{l,u}}\Big) 
\end{equation}
where $U^+ = U \cup u$ and $c_l$ is the total number of visits to $l$, and $c_{l,u}$ is the number of visits to $l$ that is contributed by user $u$. Equation (\ref{eq:entropy_add}) can be derived from Equation (\ref{eq:entropy_additivity}) if we consider $U^+$ includes two non-overlapping partitions $u$ and $U$ with associated probabilities $\frac{c_{l,u}}{c_l+c_{l,u}}$ and $\frac{c_l}{c_l+c_{l,u}}$. We note that the entropy of a single user is zero, i.e., $H(u) = 0$.

Similarly, removing a person $u$ from $U$ changes its entropy as follows:
\begin{equation}
\label{eq:entropy_remove}
\centering H(U^-) = \frac{c_l}{c_l - c_{l,u}}\bigg(H(U) - H\Big(\frac{c_{l,u}}{c_l}, \frac{c_l-c_{l,u}}{c_l}\Big)\bigg)
\end{equation}
where $U^- = U \setminus \{u\}$.

\subsection{Differential Privacy}

{\em Differential privacy (DP)}~\cite{dwork2006differential} has emerged as the de facto standard in data privacy, thanks to its strong protection guarantees rooted in statistical analysis. DP is a {\em semantic} model which provides protection against realistic adversaries with background information. Releasing data according to DP ensures that an adversary's chance of inferring any information about an individual from the sanitized data will not substantially increase, regardless of the adversary's prior knowledge. DP ensures that the adversary does not know whether an individual is present or not in the original data. DP is formally defined as follows.
\newtheorem{theorem}{Theorem}
\newtheorem{definition}{Definition}
\newtheorem{lemma}{Lemma}
\begin{definition}
\newtheorem{differential_privacy}[definition]{Definition}\label{differential_privacy}
\textsc{$\epsilon$-indistinguishability}~\cite{dwork2006calibrating}
Consider that a database produces a set of query results $\hat{D}$ on the set of queries $Q$ = $\{q_1, q_2, \ldots, q_{|Q|}\}$, and let $\epsilon>0$ be an arbitrarily small real constant. Then, transcript $U$ produced by a randomized algorithm $A$ satisfies $\epsilon$-indistinguishability if for every pair of sibling datasets $D_1$, $D_2$ that differ in only one record, it holds that
$$\ln \frac{Pr[Q(D_1) = U]}{Pr[Q(D_2) = U]} \le \epsilon$$
\end{definition}
In other words, an attacker cannot reliably learn whether the transcript was obtained by answering the query set $Q$ on dataset $D_1$ or $D_2$. Parameter $\epsilon$ is called {\em privacy budget}, and specifies the amount of protection required, with smaller values corresponding to stricter privacy protection. To achieve $\epsilon$-indistinguishability, DP injects noise into each query result, and the amount of noise required is proportional to the {\em sensitivity} of the query set $Q$, formally defined as:
\newtheorem{sensitivity}[definition]{Definition}\label{sensitivity}
\begin{sensitivity}[$L_1$-Sensitivity]~\cite{dwork2006calibrating}
Given any arbitrary sibling datasets $D_1$ and $D_2$, the sensitivity of query set $Q$ is the maximum change in their query results.
$$\sigma(Q) = \max_{D_1, D_2}||Q(D_1)-Q(D_2)||_1$$
\end{sensitivity}
An essential result from~\cite{dwork2006calibrating} shows that a sufficient condition to achieve DP with parameter $\epsilon$ is to add to each query result randomly distributed Laplace noise with mean 0 and scale $\lambda = \sigma(Q)/\epsilon$.

\section{Private Publication of LE}
\label{sec:dple}

In this section we present a baseline algorithm based on a global sensitivity of LE~\cite{dwork2006calibrating} and then introduce a thresholding technique to reduce the global sensitivity by limiting an individual's activity.

\subsection{Global Sensitivity of LE}
\label{sec:global_sen}

To achieve $\epsilon$-differential privacy, we must add noise proportional to the global sensitivity (or sensitivity for short) of LE. Thus, to minimize the amount of injected noise, we first propose a tight bound for the sensitivity of LE, denoted by $\Delta H$.
$\Delta H$ represents the maximum change of LE across all locations when the data of one user is added (or removed) from the dataset. With the following theorem, the sensitivity bound is a function of the maximum number of visits a user contributes to a location, denoted by $C$ ($C\ge 1$).
\begin{theorem}
\label{theorem:globalBoundOfDeltaH}
Global sensitivity of location entropy is
\begin{align}
\Delta H = \max\left\{\log 2, \log C - \log (\log C) - 1\right\} \nonumber
\end{align}
\end{theorem}
\begin{proof}
We prove this theorem by first deriving a tight bound for the sensitivity of a particular location $l$ (visited by $n$ users), denoted by $\Delta H_l$ (Theorem~\ref{theorem:boundOfDeltaH}). 
The bound is a function of $C$ and $n$.
Thereafter, we generalize the bound to hold for all locations as follows. We take the derivative of the bound derived for $\Delta H_l$ with respect to variable $n$ and find the extremal point where the bound is maximized.
The detailed proof 
can be found in our technical report~\cite{to2016dple}.
\end{proof}

\begin{theorem} %
\label{theorem:boundOfDeltaH}
Local sensitivity of a particular location $l$ with $n$ users is:

\begin{minipage}[t]{\linewidth}
\begin{itemize}
\item $\log 2$ when $n = 1$
\item $\log \frac{n+1}{n}$ when $C = 1$
\item %
$\max\{\log \frac{n - 1}{n - 1 + C} + \frac{C}{n - 1 + C} \log C, \log \frac{n}{n + C} + \frac{C}{n+ C} \log C,\\  \log (1 + \frac{1}{\exp(H(\mathcal{C} \setminus c_u))}) \}$
where $C$ is the maximum number of visits a user contributes to a location ($C\ge 1$)
and $H(\mathcal{C} \setminus c_u) = \log(n-1) - \frac{\log C}{C - 1} + \log \Big( \frac{\log C}{C - 1} \Big) + 1$,
when $n > 1, c > 1$.
\end{itemize}
\end{minipage}
\end{theorem}

\begin{proof}
We prove the theorem considering both cases---when a user is added (or removed) from the database. We first derive a proof for the adding case by using the additivity property of entropy from Equation~\ref{eq:entropy_add}. Similarly, the proof for the removing case can be derived from Equation~\ref{eq:entropy_remove}.
The detailed proofs can be found in 
our technical report~\cite{to2016dple}.
\end{proof}

\textbf{Baseline Algorithm:}
In this section we present a baseline algorithm that publishes location entropy for all locations (see Algorithm~\ref{alg:baseline}).
Since adding (or removing) a single user from the dataset would impact the entropy of all locations he visited, the change of adding (or removing) a user to all locations is bounded by $M_{max}\Delta H$, where $M_{max}$ is the maximum number of locations visited by a user. Thus, Line~\ref{line:add_noise} adds randomly distributed Laplace noise with mean zero and scale $\lambda = \frac{M_{max}\Delta H}{\epsilon}$ to the actual value of location entropy $H(l)$. It has been proved~\cite{dwork2006calibrating} that this is sufficient to achieve differential privacy with such simple mechanism.

\begin{algorithm} [ht]
\caption{\sc Baseline Algorithm}
\small
\begin{algorithmic}[1]
\STATE Input: privacy budget $\epsilon$, a set of locations $L=\{l_1, l_2,...,l_{|L|}\}$, maximum number of visits of a user to a location $C_{max}$, maximum number of locations a user visits $M_{max}$.
\STATE Compute sensitivity $\Delta H$ from Theorem~\ref{theorem:globalBoundOfDeltaH} for $C= C_{max}$.
\STATE For each location $l$ in $L$
\INDSTATE Count \#visits each user made to $l$: $c_{l,u}$ and compute $p_{l,u}$ \label{line:compute_p}
\INDSTATE Compute $H(l)= - \sum_{u \in U_l} p_{l,u} \log p_{l,u}$ \label{line:compute_le}
\INDSTATE Publish noisy LE: $\hat{H}(l)=H(l) + Lap(\frac{M_{max}\Delta H}{\epsilon})$ \label{line:add_noise}
\end{algorithmic}
\label{alg:baseline}
\end{algorithm}

\subsection{Reducing the Global Sensitivity of LE}

\subsubsection{Limit Algorithm}
\label{sec:limit}

\textbf{Limitation of the Baseline Algorithm:}
Algorithm~\ref {alg:baseline} provides privacy; however, the added noise is excessively high, rendering the results useless. 
To illustrate, \hyperref[fig:boundByC]{Figure~\ref{fig:boundByC}} shows the bounds of the global sensitivity (Theorem~\ref{theorem:globalBoundOfDeltaH}) when $C$ varies. The figure shows that the bound monotonically increases when $C$ grows. Therefore, the noise introduced by Algorithm~\ref{alg:baseline} increases as $C$ and $M$ increase. In practice, $C$ and $M$ can be large because a user may have visited either many locations or a single location many times, resulting in large sensitivity.
Furthermore, Figure~\ref{fig:sensitivity_m5} depicts different values of noise magnitude (in log scale) used in our various algorithms by varying the number of users visiting a location, $n$. The graph shows that the noise magnitude of the baseline is too high to be useful (see Table~\ref{tab:datasets}). 

\begin{figure}[ht]
	\begin{minipage}[b]{0.49\linewidth}
		\centering
		\includegraphics[width=1\textwidth]{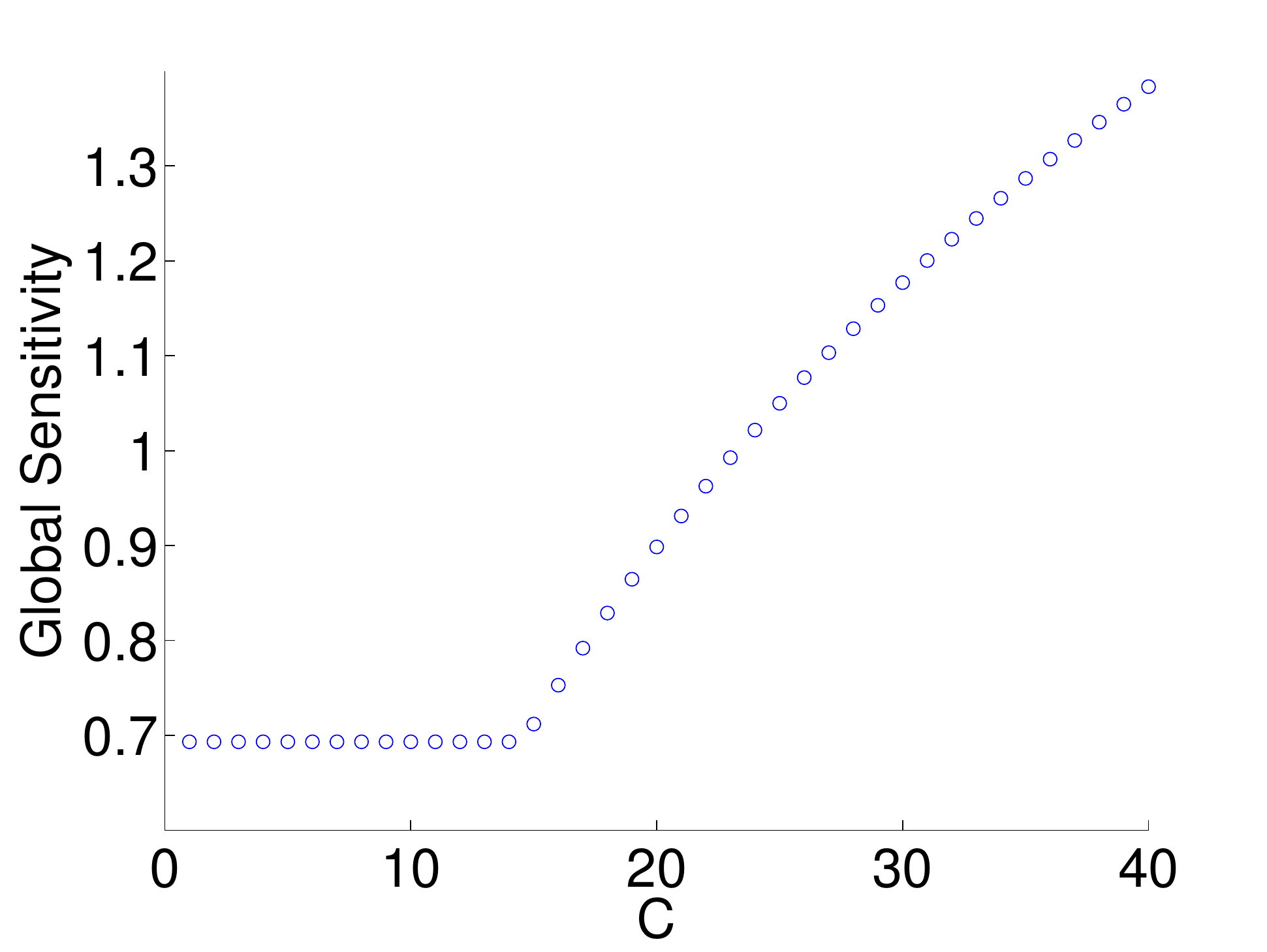}
		\caption{Global sensitivity bound of location entropy when varying $C$.}
		\label{fig:boundByC}
	\end{minipage}
		\hspace{3pt}
	\begin{minipage}[b]{0.49\linewidth}
		\centering
		\includegraphics[width=1\textwidth]{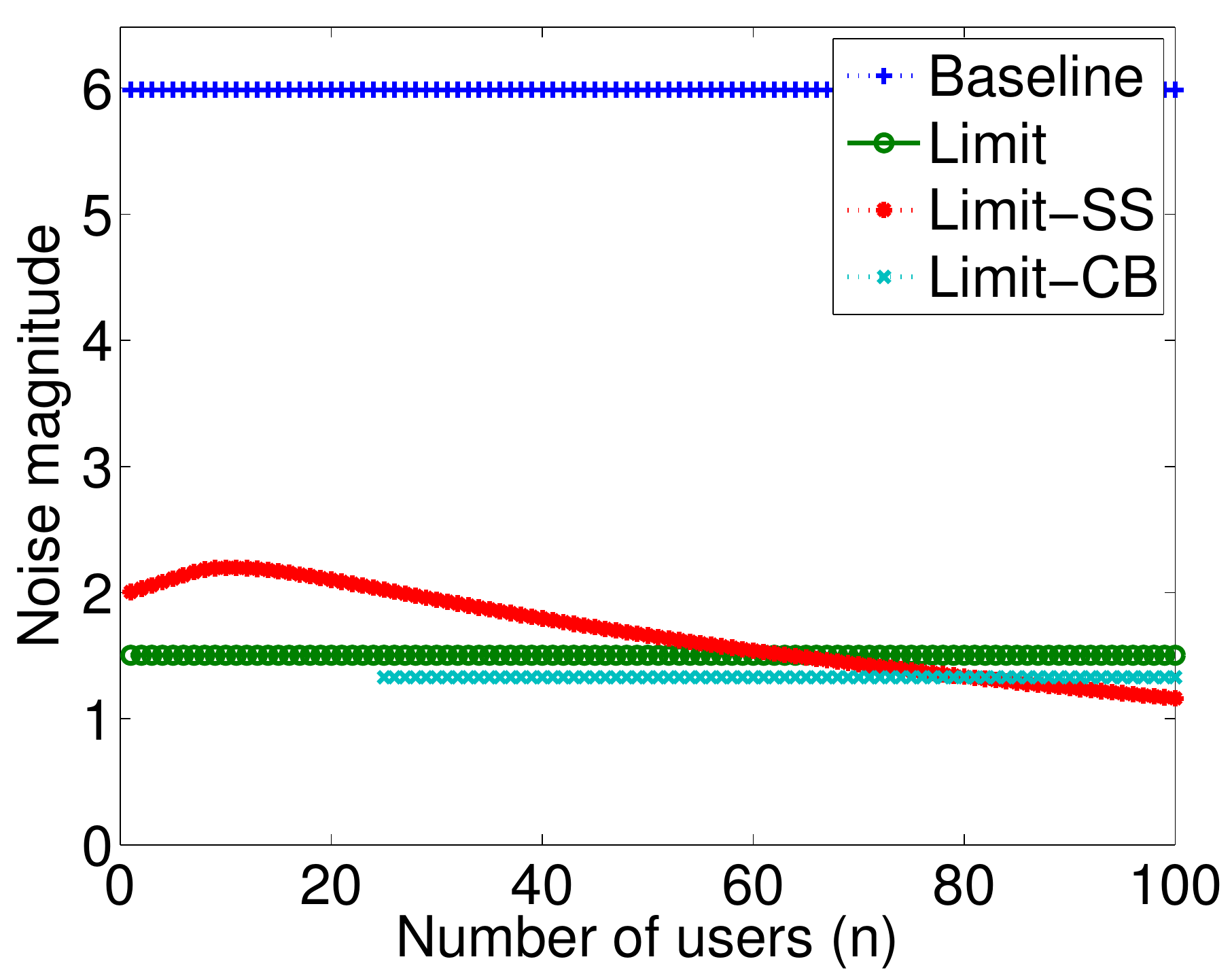}
		\caption{Noise magnitude in natural log scale ($\epsilon=5$, $C_{max}$=1000, $M_{max}$=100, $C$=20, $M$=5, $\delta$=$10^{-8}$, $k$=25).}
		\label{fig:sensitivity_m5}
	\end{minipage}
	
\end{figure}

\textbf{Improving Baseline by Limiting User Activity:}
To reduce the global sensitivity of LE, and inspired by~\cite{Korolova:2009:releasingsearch}, we propose a thresholding technique, named {\LM}, to limit an individual's activity by truncating $C$ and $M$. Our technique is based on the following two observations.
First, Figure~\ref{fig:large_visits} shows the maximum number of visits a user contributes to a location in the  Gowalla dataset that will be used in Section~\ref{sec:experiment} for evaluation. Although most users have one and only one visit, the sensitivity of LE is determined by the worst-case scenario---the maximum number of visits\footnote{\small{This suggests that users tend not to check-in at places that they visit the most, e.g., their homes, because if they did, the peak of the graph would not be at 1.}}.
Second, Figure~\ref{fig:many_locs} shows the number of locations visited by a user. The figure confirms that there are many users who contribute to more than ten locations.

\begin{figure}[ht]
	\begin{minipage}[b]{.49\linewidth}
		\centering
		\includegraphics[width=1\textwidth]{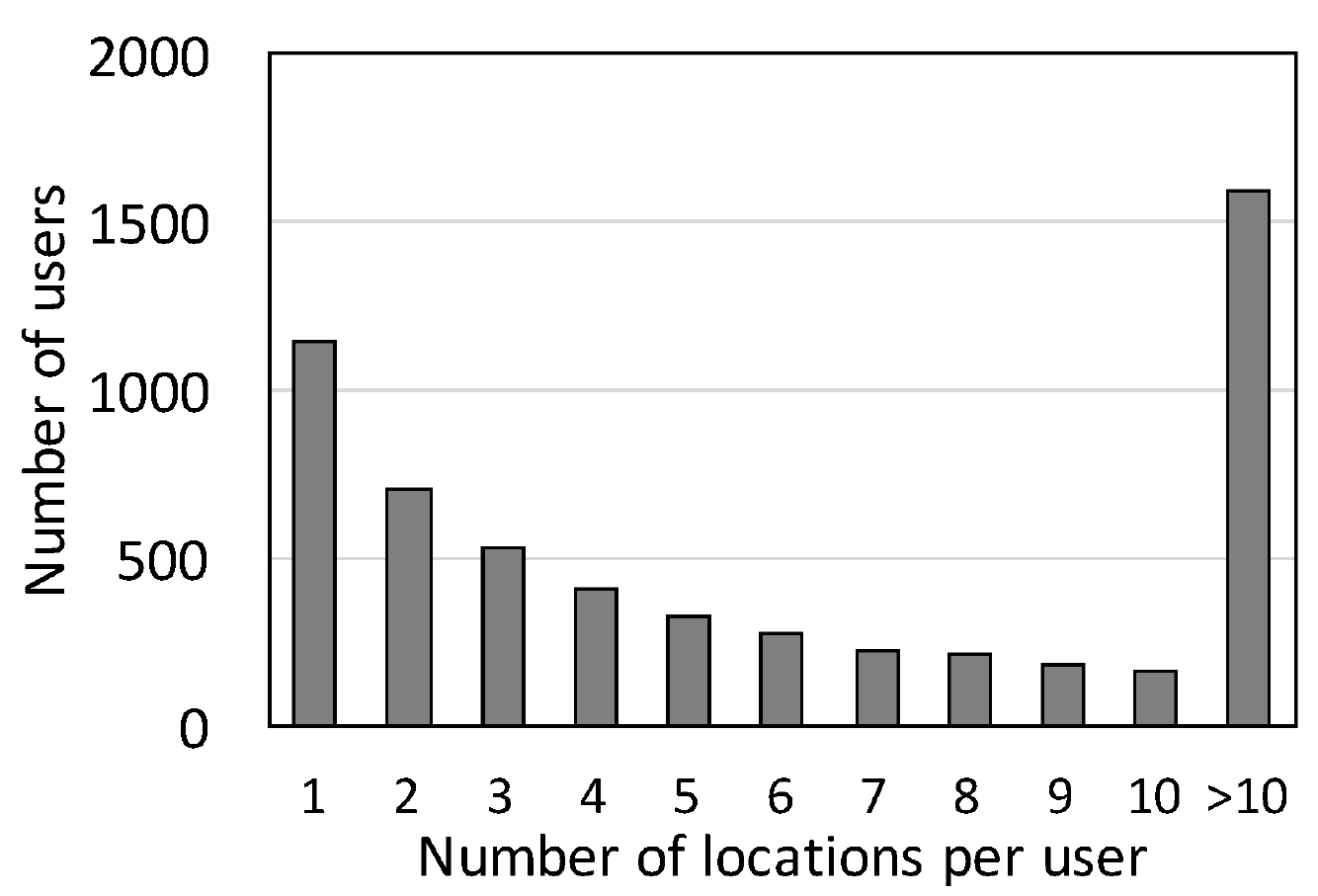}
		\subcaption{A user may visit many locations}
		\label{fig:many_locs}
	\end{minipage}
	\hspace{5pt}
	\begin{minipage}[b]{.49\linewidth}
		\centering
		\includegraphics[width=1\textwidth]{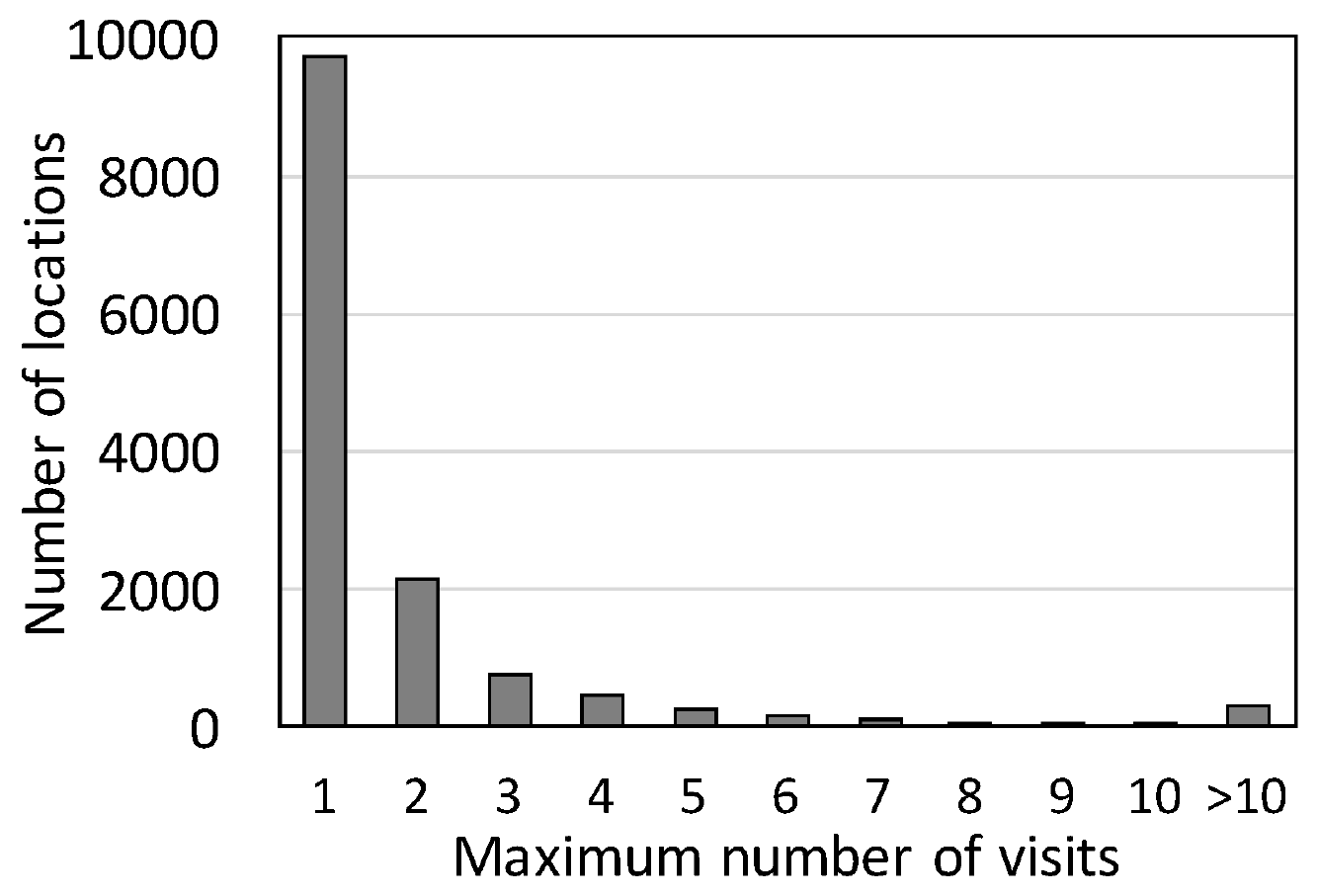}
		\subcaption{The largest number of visits a user contributes to a  location}
		\label{fig:large_visits}
	\end{minipage}
	\caption{Gowalla, New York.}
	\label{fig:gowalla_ny_stats}
	\vspace{-10pt}
\end{figure}

Since the introduced noise linearly increases with $M$ and monotonically increases with $C$, the noise can be reduced by capping them.
First, to truncate $M$, we keep the \emph{first} $M$ location visits of the users who visit more than $M$ locations and throw away the rest of the locations' visits. As a result, adding or removing a single user in the dataset affects at most $M$ locations.
Second, we set the number of visits of the users who have contributed more than $C$ visits to a particular location of $C$.
Figure~\ref{fig:sensitivity_m5} shows that the noise magnitude used in {\LM} drops by two orders of magnitude when compared with the baseline's sensitivity.

At a high-level, {\LM} (Algorithm~\ref{alg:limit}) works as follows.
Line~\ref{line:lim1} limits user activity across locations, while Line~\ref{line:lim2} limits user activity to a location. The impact of Line~\ref{line:lim1} is the introduction of approximation error on the published data. This is because the number of users visiting some locations may be reduced, which alters their actual LE values. Subsequently, some locations may be thrown away without being published.
Furthermore, Line~\ref{line:lim2} also alters the value of location entropy, but by trimming the number of visits of a user to a location.
The actual LE value of location $l$ (after thresholding $M$ and $C$) is computed in Line~\ref{line:le_lim}. Consequently, the noisy LE is published in Line~\ref{line:lap_lim}, where $Lap(\frac{M\Delta H}{\epsilon})$ denotes a random variable drawn independently from Laplace distribution with mean zero and scale parameter $\frac{M\Delta H}{\epsilon}$.

\begin{algorithm} [ht]
\caption{\sc {\LM} Algorithm}
\small
\begin{algorithmic}[1]
\STATE Input: privacy budget $\epsilon$, a set of locations $L=\{l_1, l_2,...,l_{|L|}\}$, maximum threshold on the number of visits of a user to a location $C$, maximum threshold on the number of locations a user visits $M$
\STATE For each user $u$ in $U$ \label{line:for_lim}
\INDSTATE Truncate $M$: keep the first $M$ locations' visits of the users who visit more than $M$ locations \label{line:lim1}
\STATE Compute sensitivity $\Delta H$ from Theorem~\ref{theorem:globalBoundOfDeltaH}.
\STATE For each location $l$ in $L$
\INDSTATE Count \#visits each user made to $l$: $c_{l,u}$ and compute $p_{l,u}$
\INDSTATE Threshold $C$: $\bar{c}_{l,u}=\min(C,c_{l,u})$, then compute $\bar{p}_{l,u}$ \label{line:lim2}
\INDSTATE Compute $\bar{H}(l)= - \sum_{u \in U_l} \bar{p}_{l,u} \log \bar{p}_{l,u}$ \label{line:le_lim}
\INDSTATE Publish noisy LE: $\hat{H}(l)=\bar{H}(l) + Lap(\frac{M\Delta H}{\epsilon})$ \label{line:lap_lim}
\end{algorithmic}
\label{alg:limit}
\end{algorithm}

The performance of Algorithm~\ref{alg:limit} depends on how we set $C$ and $M$. There is a trade-off on the choice of values for $C$ and $M$. Small values of $C$ and $M$ introduce small perturbation error but large approximation error and vice versa. Hence, in Section~\ref{sec:experiment}, we empirically find the values of $M$ and $C$  that strike a balance between noise and approximation errors.

\subsubsection{Privacy Guarantee of the Limit Algorithm}
\label{sec:proof_baseline}
The following theorem shows that Algorithm~\ref{alg:limit} is differentially private.
\begin{theorem} %
\label{theorem:privacyGuarantee}
\hyperref[alg:limit]{Algorithm~\ref{alg:limit}} satisfies $\epsilon$-differential privacy.
\end{theorem}

\begin{proof}
For all locations, let $L_1$ be any subset of $L$.
Let $T = \{t_1, t_2, \dots, t_{|L_1|}\} \in $\textit{Range}($\mathcal{A}$) denote an arbitrary possible output.
Then we need to prove the following:
\begin{displaymath}
\begin{split}
& \quad \frac{Pr[\mathcal{A}(O_{1\text{(org)}}, \dots, O_{|L_1|\text{(org)}}) = T]}{Pr[\mathcal{A}(O_{1\text{(org)}} \setminus O_{l, u\text{(org)}}, \dots, O_{|L_1|\text{(org)}} \setminus O_{l, u\text{(org)}}) = T]} \\
& \le exp(\epsilon) \\
\end{split}
\end{displaymath}

The details of the proof and notations used can be found in 
our technical report~\cite{to2016dple}.
\end{proof}

\section{Relaxation of Private LE}
\label{sec:relaxation}

This section presents our utility enhancements by adopting two weaker notions of privacy: smooth sensitivity~\cite{Nissim:2007:SmoothSensitivity} (slightly weaker) and crowd-blending~\cite{gehrke2012crowd} (strictly weaker).

\subsection{Relaxation with Smooth Sensitivity}
\label{sec:smoothsensitivity}

We aim to extend {\LM} to publish location entropy with smooth sensitivity (or SS for short). 
We first present the notions of smooth sensitivity and the {\LMS} algorithm. We then show how to precompute the SS of location entropy.

\subsubsection{{\LMS} Algorithm}

Smooth sensitivity is a technique that allows one to compute noise magnitude---not only by the function one wants to release (i.e., location entropy), but also by the database itself.
The idea is to use the local sensitivity bound of each location rather than the global sensitivity bound, resulting in small injected noise.
However, simply adopting the local sensitivity to calibrate noise may leak the information about the number of users visiting that location.
Smooth sensitivity is stated as follows.

Let $x, y \in D^N$ denote two databases, where $N$ is the number of users.
Let $l^x, l^y$ denote the location $l$ in database $x$ and $y$, respectively.
Let $d(l^x, l^y)$ be the Hamming distance between $l^x$ and $l^y$,
which is the number of users at location $l$ on which $x$ and $y$ differ; i.e., $d(l^x, l^y)=|\{i : l^x_i \neq l^y_i \}|$; $l^x_i$ represents information contributed by one individual.
The local sensitivity of location $l^x$, denoted by $LS(l^x)$, is the maximum change of location entropy when a user is added or removed.
\begin{definition}{Smooth sensitivity~\cite{Nissim:2007:SmoothSensitivity}}
For $\beta > 0$, $\beta$-smooth sensitivity of location entropy is:
\begin{displaymath}
\begin{split}
SS_{\beta}(l^x)	&= \max_{l^y \in D^N} \Big( LS(l^y) \cdot e^{-\beta d(l^x, l^y)} \Big) \\
						&= \max_{k=0,1,\dots,N} e^{-k\beta} \Big( \max_{y: d(l^x,l^y) = k} LS(l^y) \Big)\\
\end{split}
\end{displaymath}
\end{definition}
Smooth sensitivity of LE of location $l^x$ can be interpreted as the maximum of $LS(l^x)$ and $LS(l^y)$ where the effect of $y$ at distance $k$ from $x$ is dropped by a factor of $e^{-k\beta}$.
Thereafter, the smooth sensitivity of LE can be plugged into Line~\ref{line:lap} of Algorithm~\ref{alg:limit}, producing the {\LMS} algorithm.
\begin{algorithm} [ht]
\caption{\sc {\LMS} Algorithm}
\small
\begin{algorithmic}[1]
\STATE Input: privacy budget $\epsilon$, privacy parameter $\delta$, $L=\{l_1, l_2,...,l_{|L|}\}$, $C, M$
\STATE Copy Lines~\ref{line:for_lim}-\ref{line:le_lim} from Algorithm~\ref{alg:limit}
\INDSTATE Publish noisy LE $\hat{H}(l)=\bar{H}(l) + \frac{M \cdot 2 \cdot SS_{\beta}(l)}{\epsilon} \cdot \eta$,
	where $\eta \sim \text{Lap(1)}$, where $\beta = \frac{\epsilon}{2\ln(\frac{2}{\delta})}$ \label{line:lap}
\end{algorithmic}
\label{alg:limit-ss}
\end{algorithm}

\subsubsection{Privacy Guarantee of {\LMS}}
The noise of {\LMS} is specific to a particular location as opposed to those of the \hyperref[alg:baseline]{{\BL}} and \hyperref[alg:limit]{{\LM}} algorithms.
{\LMS} has a slightly weaker privacy guarantee. It satisfies $(\epsilon, \delta)$-differential privacy, where $\delta$ is a privacy parameter, $\delta = 0$ in the case of Definition~\ref{differential_privacy}. The choice of $\delta$ is generally left to the data releaser. Typically, $\delta < \frac{1}{\text{number of users}}$ (see~\cite{Nissim:2007:SmoothSensitivity}
for details).

\begin{theorem}{Calibrating noise to smooth sensitivity~\cite{Nissim:2007:SmoothSensitivity}}
\label{cor:addNoiseToSmoothSensitivity}
If $\beta \le \frac{\epsilon}{2\ln(\frac{2}{\delta})}$ and $\delta \in (0, 1)$, the algorithm $l \mapsto H(l) + \frac{2 \cdot SS_{\beta}(l)}{\epsilon} \cdot \eta$,
	where $\eta \sim \text{Lap(1)}$, is $(\epsilon, \delta)$-differentially private.
\end{theorem}

\begin{theorem} 
\label{theorem:privacyLimitSS}
{\LMS} is $(\epsilon, \delta)$-differentially private.
\end{theorem}

\begin{proof}
Using \hyperref[cor:addNoiseToSmoothSensitivity]{Theorem~\ref{cor:addNoiseToSmoothSensitivity}}, 
$\mathcal{A}_l$ satisfies $(0)$-differential privacy when $l \notin L_1 \cap L(u)$, and satisfies $(\frac{\epsilon}{M}, \frac{\delta}{M})$-differential privacy when $l \in L_1 \cap L(u)$.
\end{proof}

\subsubsection{Precomputation of Smooth Sensitivity}
\label{subsec:computerSmoothSensitivity}

This section shows that the smooth sensitivity of a location visited by $n$ users can be effectively precomputed.
\hyperref[fig:sensitivity_m5]{Figure~\ref{fig:sensitivity_m5}} illustrates the precomputed local sensitivity for a fixed value of $C$. 

Let $LS(C, n), SS(C, n)$ be the local sensitivity and the smooth sensitivity of all locations that visited by $n$ users, respectively. $LS(C, n)$ is defined in \hyperref[theorem:boundOfDeltaH]{Theorem~\ref{theorem:boundOfDeltaH}}.
Let $GS(C)$ be the global sensitivity of the location entropy given $C$, which is defined in \hyperref[theorem:globalBoundOfDeltaH]{Theorem~\ref{theorem:globalBoundOfDeltaH}}.
\hyperref[alg:compute_ss]{Algorithm~\ref{alg:compute_ss}} computes $SS(C, n)$.
At a high level, the algorithm computes the effect of all locations at every possible distance $k$ from $n$, which is non-trivial. Thus, to speed up computations, we propose two stopping conditions based on the following observations.

Let $n^x, n^y$ be the number of users visited $l^x, l^y$, respectively.
If $n^x > n^y$, Algorithm~\ref{alg:compute_ss} stops when $e^{-k\beta} GS(C)$ is less than the current value of smooth sensitivity (\hyperref[line:stopSmallN]{Line~\ref{line:stopSmallN}}).
If $n^x < n^y$, given the fact that $LS(l^y)$ starts to decrease when $n^y > \frac{C}{\log C - 1} + 1$, and $e^{-k\beta}$ also decreases when $k$ increases, 
Algorithm~\ref{alg:compute_ss} also terminates when $n^y > \frac{C}{\log C - 1} + 1$ (\hyperref[line:stopSmallN]{Line~\ref{line:stopBigN}}).
In addition, the algorithm tolerates a small value of smooth sensitivity $\xi$. 
Thus, when $n$ is greater than $n_0$ such that $LS(C, n_0) < \xi$, the precomputation of $SS(C, n)$ is stopped and $SS(C, n)$ is considered as $\xi$ for all $n > n_0$ (\hyperref[line:stopSmallN]{Line~\ref{line:stopBigN}}).

\begin{algorithm} [ht]
\caption{\sc Precompute Smooth Sensitivity}
\small
\begin{algorithmic}[1]
\STATE Input: privacy parameters: $\epsilon, \delta, \xi$; $C$, maximum number of possible users $N$
\STATE Set $\beta = \frac{\epsilon}{2\ln(\frac{2}{\delta})}$
\STATE For $n = [1, \dots, N]$
\INDSTATE $SS(C, n) = 0$
\INDSTATE For $k = [1,\dots,N]$:
\INDSTATE[2] $SS(C, n) = \max \Big(SS(C, n), e^{-k\beta} \max (LS(C, n-k), LS(C, n+k)) \Big)$
\INDSTATE[2] Stop when $e^{-k\beta} GS(C, n-k) < SS(C, n)$ and $n + k > \frac{C}{\log C - 1} + 1$ \label{line:stopSmallN}
\INDSTATE Stop when $n > \frac{C}{\log C - 1} + 1$ and $LS(C, n) < \xi$ \label{line:stopBigN}
\end{algorithmic}
\label{alg:compute_ss}
\end{algorithm}

\subsection{Relaxation with Crowd-Blending Privacy}

\subsubsection{{\LMC} Algorithm}

Thus far, we publish entropy for all locations; however, the ratio of noise to the true value of LE (noise-to-true-entropy ratio) is often excessively high when the number of users visiting a location $n$ is small (i.e., Equation~\ref{eq:entropy_maximum} shows that entropy of a location is bounded by $\log(n)$). The large noise-to-true-entropy ratio would render the published results useless since the introduced noise outweighs the actual value of LE.
This is an inherent issue with the sparsity of the real-world datasets. For example, Figure~\ref{fig:sparsity} summarizes the number of users contributing visits to each location in the Gowalla dataset. The figure shows that most locations have check-ins from fewer than ten users. These locations have LE values of less than $\log(10)$, which are particularly prone to the noise-adding mechanism in differential privacy.

\begin{figure}[ht]
	\begin{minipage}[b]{0.48\linewidth}
		\centering
  \includegraphics[width=1\textwidth]{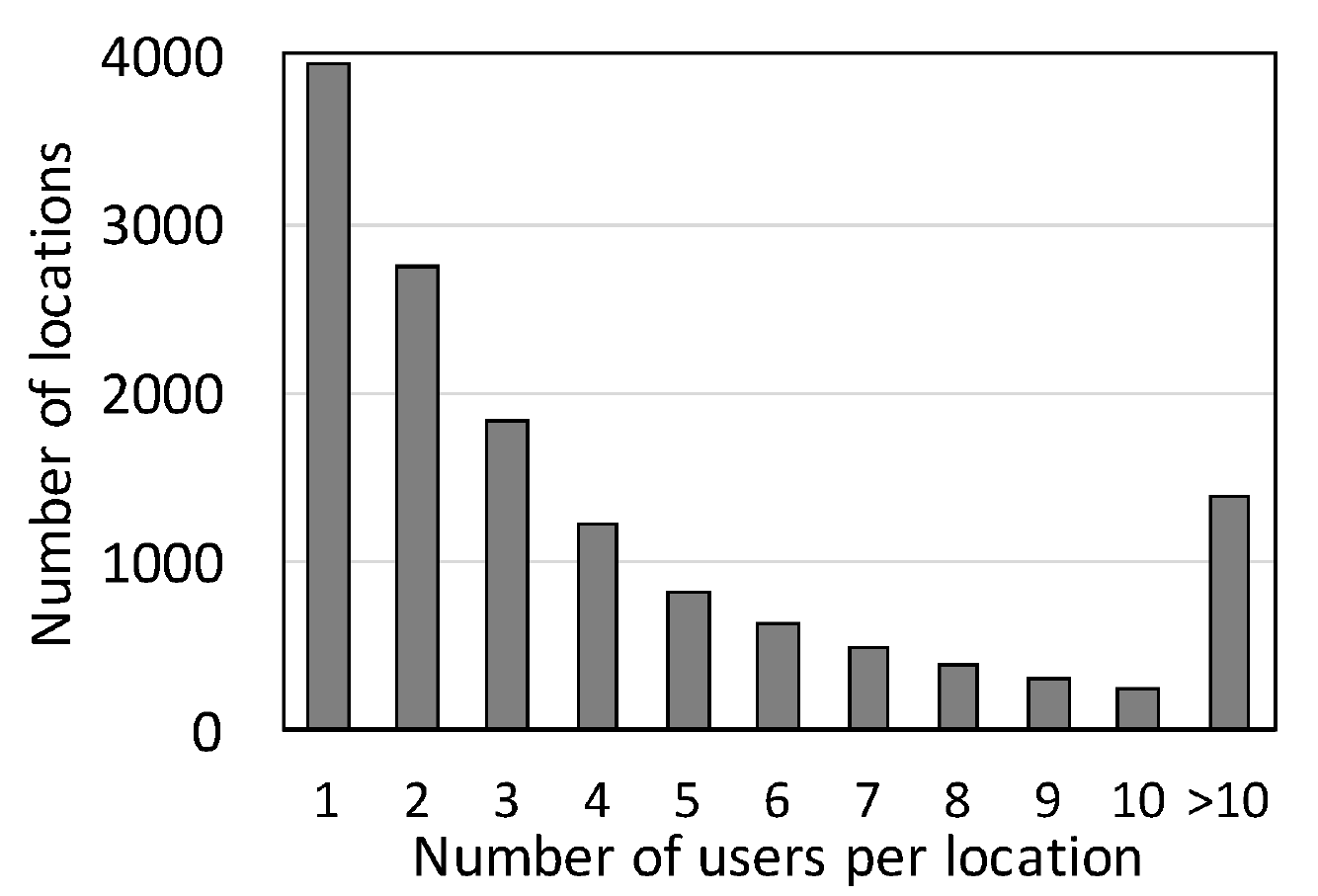}
  \caption{Sparsity of location visits (Gowalla, New York).}
  \label{fig:sparsity}
	\end{minipage}
	\vspace{-5pt}
	\hspace{2pt}
	\begin{minipage}[b]{0.48\linewidth}
		\centering
  \includegraphics[width=1\textwidth]{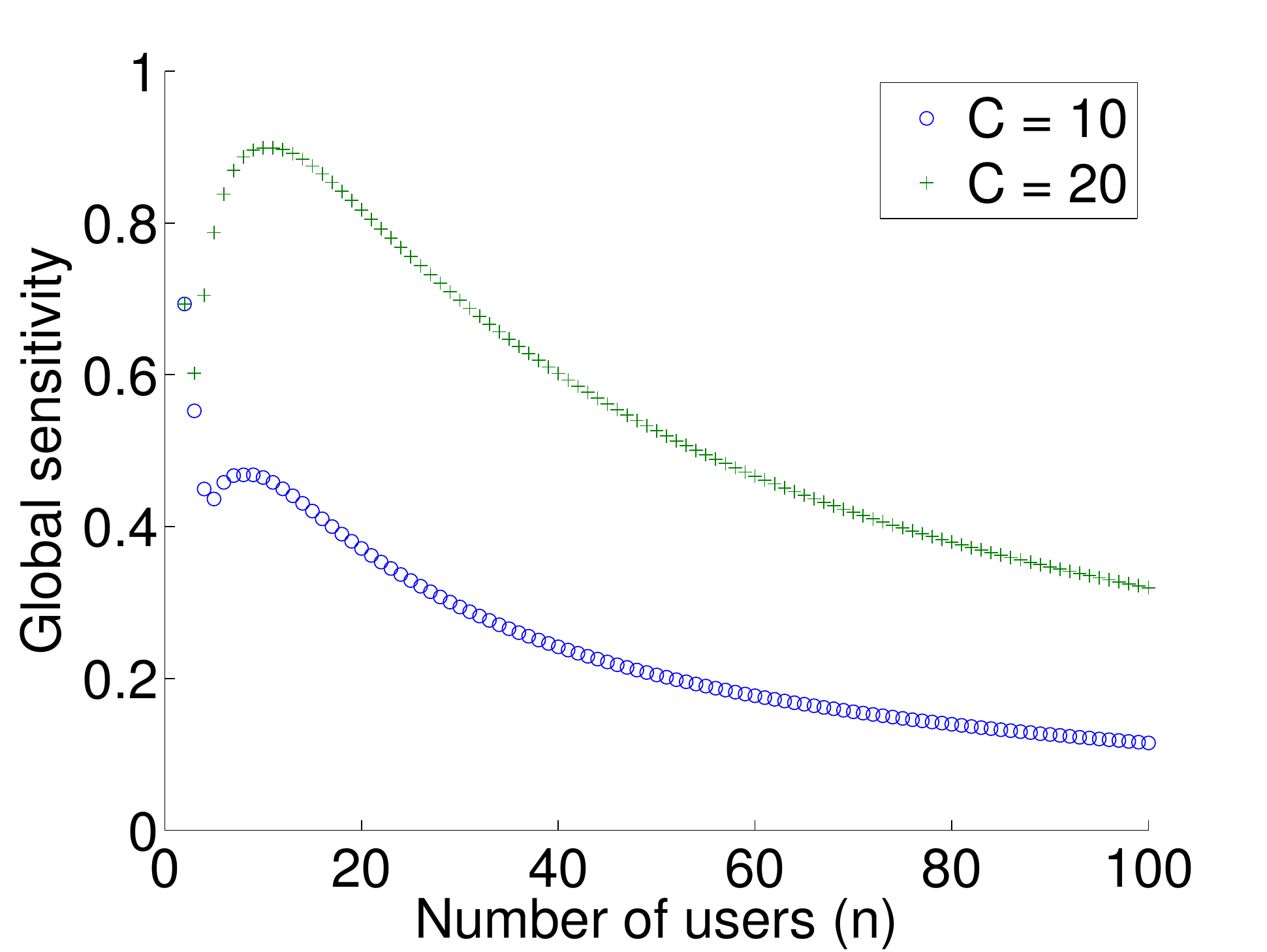}
  \caption{Global sensitivity bound when varying $n$.}
  \label{fig:boundByVaryNFixedC}
	\end{minipage}
	\vspace{-5pt}
\end{figure}

Therefore, to reduce the noise-to-true-entropy ratio, we propose a small sensitivity bound of location entropy that depends on the minimum number of users visiting a location, denoted by $k$. Subsequently, we present Algorithm~\ref{alg:crowd} that satisfies $(k,\epsilon)$-crowd-blending privacy~\cite{gehrke2012crowd}. We prove this in Section~\ref{sec:proof_crowd}.

The algorithm aims to publish entropy of locations with at least $k$ users ($n \ge k$) and throw away the other locations. We refer to the algorithm as {\LMC}.
Lines~\ref{line:start_cb}-\ref{line:end_cb} publish the entropy of each location according to ($k,\epsilon$)-crowd-blending privacy. That is, we publish the entropy of the locations with at least $k$ users and suppress the others. The following lemma shows that for the locations with at least $k$ users we have a tighter bound on $\Delta H$, which depends on $C$ and $k$.
Figure~\ref{fig:sensitivity_m5} shows that the sensitivity used in {\LMC} is significantly smaller than {\LM}'s sensitivity.

\begin{theorem}
\label{theorem:globalBoundOfDeltaH_limit_n}
Global sensitivity of location entropy for locations with at least $k$ users, $k \ge \frac{C}{\log C - 1} + 1$, where $C$ is the maximum number of visits a user contributes to a location, is the local sensitivity at $n = k$.

\end{theorem}

\begin{proof}
We prove the theorem by showing that local sensitivity decreases when the number of users $n \ge \frac{C}{\log C - 1} + 1$.
Thus, when $n \ge \frac{C}{\log C - 1} + 1$, the global sensitivity equals to the local sensitivity at the smallest value of $n$, i.e, $n = k$.
The detailed proof can be found in 
our technical report~\cite{to2016dple}.
\end{proof}

\begin{algorithm} [ht]
\caption{\sc {\LMC} Algorithm}
\small
\begin{algorithmic}[1]
\STATE Input: all users $U$, privacy budget $\epsilon$; $C, M, k$
\INDSTATE Compute global sensitivity $\Delta H$ based on Theorem~\ref{theorem:globalBoundOfDeltaH_limit_n}.
\STATE For each location $l\in L$	\label{line:start_cb}
\INDSTATE Count number of users who visit $l$, $n_l$
\INDSTATE If $n_l \ge k$, publish $\hat{H}(l)$ according to Algorithm~\ref{alg:limit} with budget $\epsilon$ using a tighter bound on $\Delta H$  \label{line:lap_cb}
\INDSTATE Otherwise, do not publish the data	\label{line:end_cb}
\end{algorithmic}
\label{alg:crowd}
\end{algorithm}

\subsubsection{Privacy Guarantee of {\LMC}}
\label{sec:proof_crowd}

Before proving the privacy guarantee of {\LMC}, we first present the notion of crowd-blending privacy, a strict relaxation of differential privacy~\cite{gehrke2012crowd}. $k$-crowd blending private sanitization of a database requires each individual in the database to blend with $k$ other individuals in the database. This concept is related to k-anonymity~\cite{sweeney2002k} since both are based on the notion of ``blending in a crowd." However, unlike k-anonymity that only restricts the published data, crowd-blending privacy imposes restrictions on the noise-adding mechanism. 
Crowd-blending privacy is defined as follows.
\newtheorem{crowd_blending}[definition]{Definition}\label{crowd_blending}
\begin{crowd_blending}[Crowd-blending privacy]
An algorithm $A$ is $(k, \epsilon)$-crowd-blending private if for every database $D$ and every individual $t \in D$, either $t$ $\epsilon$-blends in a crowd of $k$ people in $D$, or $A(D) \approx_\epsilon A(D \backslash \{t\})$ (or both).
\end{crowd_blending}
A result from~\cite{gehrke2012crowd} shows that differential privacy implies crowd-blending privacy.
\begin{theorem}
\newtheorem{diff_crowd}[definition]{Theorem}\label{diff_crowd}
\textsc{DP $\longrightarrow$ Crowd-blending privacy}~\cite{dwork2006calibrating}
Let $A$ be any $\epsilon$-differentially private algorithm. Then, $A$ is ($k,\epsilon$)-crowd-blending private for every integer $k \ge 1$.
\end{theorem}
The following theorem shows that Algorithm~\ref{alg:crowd} is ($k,\epsilon$)-crowd-blending private.
\begin{theorem} 
\label{crowd_blend}
\hyperref[alg:crowd]Algorithm~\ref{alg:crowd} is ($k,\epsilon$)-crowd-blending private.
\end{theorem}
\begin{proof}
First, if there are at least $k$ people in a location, then individual $u$ $\epsilon$-blends with $k$ people in $U$. This is because Line~\ref{line:lap_cb} of the algorithm satisfies $\epsilon$-differential privacy, which infers ($k,\epsilon$)-crowd-blending private (Theorem~\ref{diff_crowd}).
Otherwise, we have $A(D) \approx_0 A(D \backslash \{t\})$ since $A$ suppresses each location with less than $k$ users.
\end{proof}

\section{Performance Evaluation}
\label{sec:experiment}
We conduct several experiments on real-world and synthetic datasets to compare the effectiveness and utility of our proposed algorithms. Below, we first discuss our experimental setup. Next, we present our experimental results.

\subsection{Experimental Setup}\label{s:evaluation}
\label{sec:setup}

\textbf{Datasets}: We conduct experiments on one real-world (Gowalla) and two synthetic datasets (Sparse and Dense).
The statistics of the datasets are shown in Table~\ref{tab:datasets}.
Gowalla contains the check-in history of users in a location-based social network. For our experiments, we use the check-in data in an area covering the city of New York. %

\begin{table}
\begin{center}
\footnotesize
\begin{tabular}{ l | c | c | r }
\hline
 & \textbf{Sparse} & \textbf{Dense} & \textbf{Gow.}  \tn
\hline
\# of locations & 10,000 & 10,000 & 14,058  \tn
\hline
\# of users & 100K & 10M & 5,800 \tn
\hline
Max LE & 9.93 & 14.53  & 6.45 \tn
\hline
Min LE & 1.19 & 6.70  & 0.04 \tn
\hline
Avg. LE & 3.19 & 7.79 & 1.45 \tn
\hline
Variance of LE & 1.01 & 0.98 & 0.6  \tn
\hline
Max \#locations per user & 100 & 100 & 1407  \tn
\hline
Avg. \#locations per user & 19.28 & 19.28  & 13.5 \tn
\hline
Max \#visits to a loc. per user & 20,813 & 24,035 & 162  \tn
\hline
Avg. \#visits to a loc. per user & 2578.0 & 2575.8 & 7.2  \tn
\hline
Avg. \#users per loc. & 192.9 & 19,278 & 5.6  \tn
\hline
\end{tabular}
\caption{Statistics of the datasets.}
\label{tab:datasets}
\end{center}
\end{table}
For synthetic data generation, in order to study the impact of the density of the dataset, we consider two cases: Sparse and Dense. Sparse contains 100,000 users while Dense has 10 million users.
The Gowalla dataset is sparse as well.  We add the Dense synthetic dataset to emulate the case for large industries, such as Google, who have access to large- and fine-granule user location data.
To generate visits, without loss of generality, the location with id $x\in[1,2,\dots,10,000]$ has a probability  $1/x$ of being visited by a user. This means that locations with smaller ids tend to have higher location entropy since more users would visit these locations.
In the same fashion, the user with id $y\in\{1,2,\dots,100,000\}$ (Sparse) is selected with probability $1/y$.
This follows the real-world characteristic of location data where a small number of locations are very popular and then many locations have a small number of visits.

In all of our experiments, we use five values of privacy budget $\epsilon \in \{0.1, 0.5, 1, \textbf{5}, 10\}$. We vary the maximum number of visits a user contributes to a location $C \in \{1,2, \dots, \textbf{5}, \dots,50\}$ and the maximum number of locations a user visits $M \in \{1,2,\textbf{5},10,20,30\}$. We vary threshold $k \in \{10, 20, 30, 40$, $\textbf{50}\}$. 
We also set $\xi = 10^{-3}, \delta = 10^{-8}$, and $\beta \approx \epsilon / {2 * ln(2 / \delta)}$.
Default values are shown in boldface.

\textbf{Metrics:} 
We use KL-divergence as one measure of preserving the original data distribution after adding noise. 
Given two discrete probability distributions $P$ and $Q$, the KL-divergence of $Q$ from $P$ is defined as follows:
\begin{equation}
D_{KL}(P || Q) = \sum_i P(i) \log \frac{P(i)}{Q(i)}
\end{equation}
In this paper the location entropy of location $l$ is the probability that $l$ is chosen when a location is randomly selected from the set of all locations; $P$ and $Q$ are respectively the published and the actual LE of locations after normalization; i.e., normalized values must sum to unity.

We also use mean squared error (MSE) over a set of locations $L$ as the metric of accuracy using Equation~\ref{eq:rmse}.
\begin{align}
\label{eq:rmse}
\mathit{MSE}=\frac{1}{|L|}\sum\limits_{l \in L}{}\big(LE_a(l) - LE_n(l)\big)^2
\end{align}
where $LE_a(l)$ and $LE_n(l)$ are the actual and noisy entropy of the location $l$, respectively.

Since {\LMC} discards more locations as compared to {\LM} and {\LMS}, we consider both cases: 1) KL-divergence and MSE metrics are computed on all locations $L$, where the entropy of the suppressed locations are set to zero (default case); 2) the metrics are computed on the subset of locations that {\LMC} publishes, termed \textit{Throwaway}.

\subsection{Experimental Results}\label{sec:result}
\label{sec:results}

We first evaluate our algorithms on the synthetic datasets.
\subsubsection{Overall Evaluation of the Proposed Algorithms}

We evaluate the performance of {\LM} from Section~\ref{sec:limit} and its variants ({\LMS} and {\LMC}).
We do not include the results for {\BL} since the excessively high amount of injected noise renders the perturbed data useless.

Figure~\ref{fig:raw} illustrates the distributions of noisy vs. actual LE on Dense and Sparse.
The actual distributions of the dense (Figure~\ref{fig:d-actual}) and sparse (Figure~\ref{fig:s-actual}) datasets confirm our method of generating the synthetic datasets; locations with smaller ids have higher entropy, and entropy of locations in Dense are higher than that in Sparse.
We observe that {\LMS} generally performs best in preserving the original data distribution for Dense (Figure~\ref{fig:d-limit-ss}), while {\LMC} performs best for Sparse (Figure~\ref{fig:s-limit-cb}). Note that as we show later, {\LMC} performs better than {\LMS} and {\LM} given a small budget $\epsilon$ (see Section~\ref{sec:vary_eps}).

\begin{figure*}[ht]
	\begin{minipage}[b]{.245\linewidth}
		\centering
		\includegraphics[width=1\textwidth]{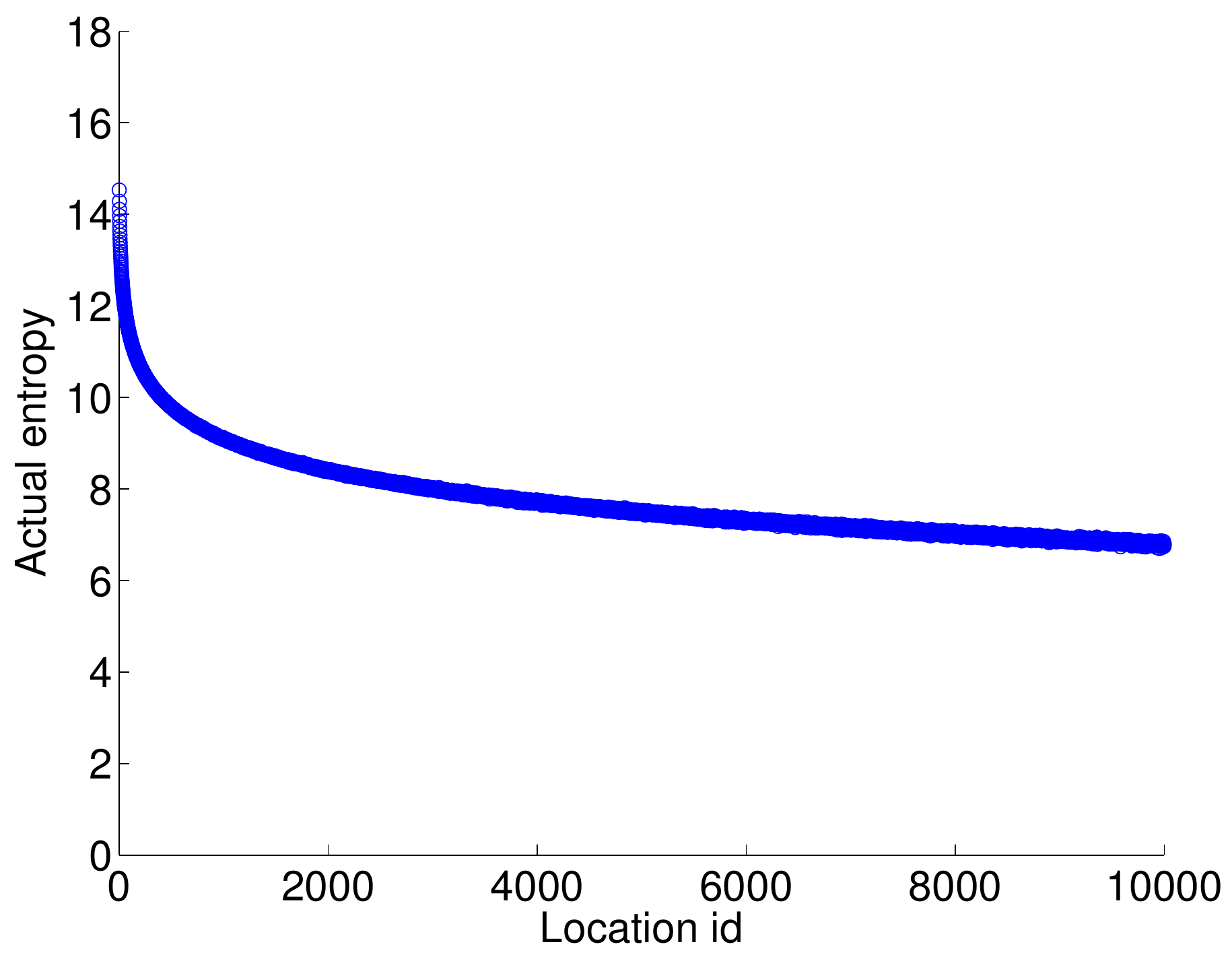}
		\subcaption{Actual (Dense)}
		\label{fig:d-actual}
	\end{minipage}
	\begin{minipage}[b]{.245\linewidth}
		\centering
		\includegraphics[width=1\textwidth]{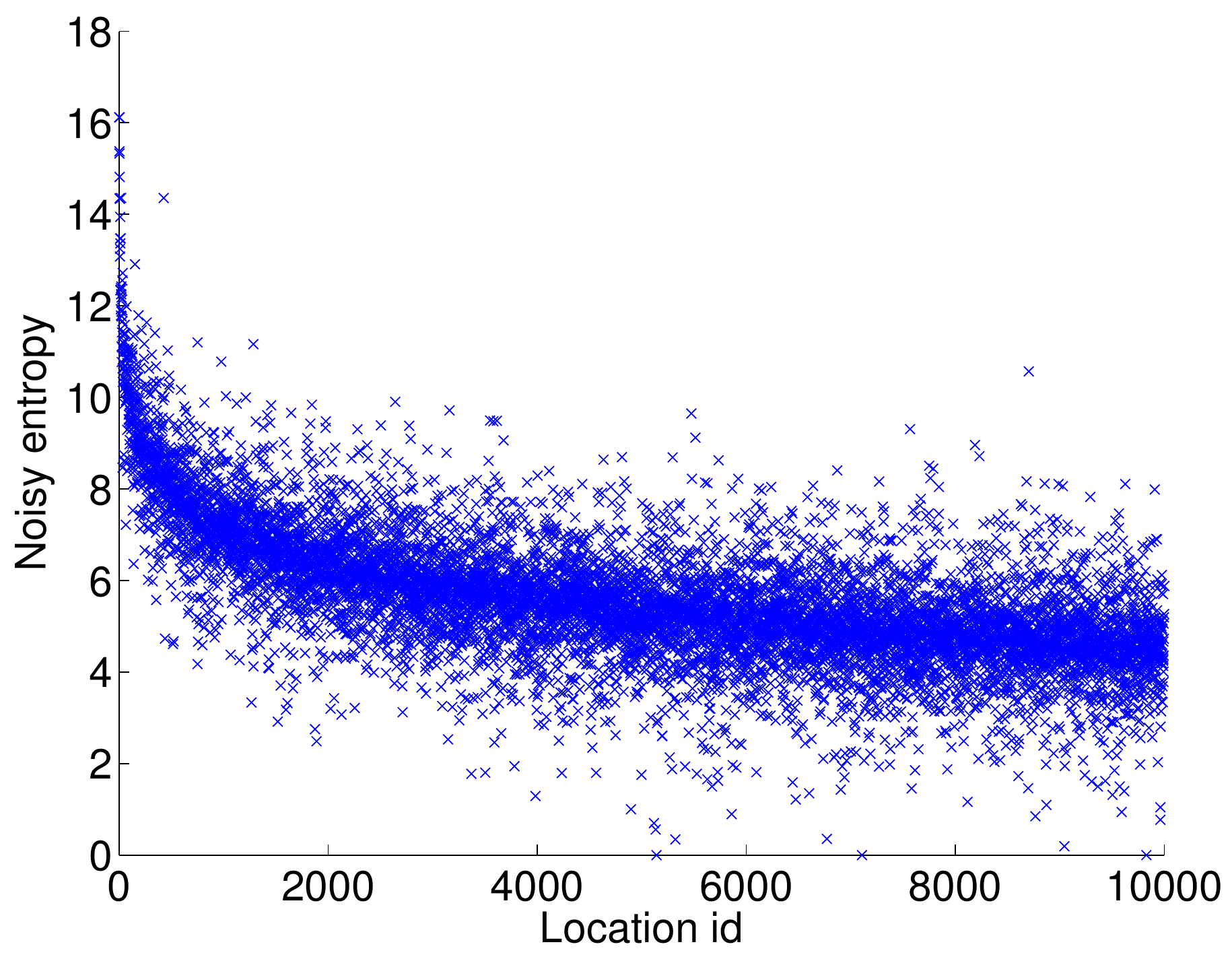}
		\subcaption{{\LM} (Dense)}
		\label{fig:d-limit}
	\end{minipage}
	\begin{minipage}[b]{.245\linewidth}
		\centering
		\includegraphics[width=1\textwidth]{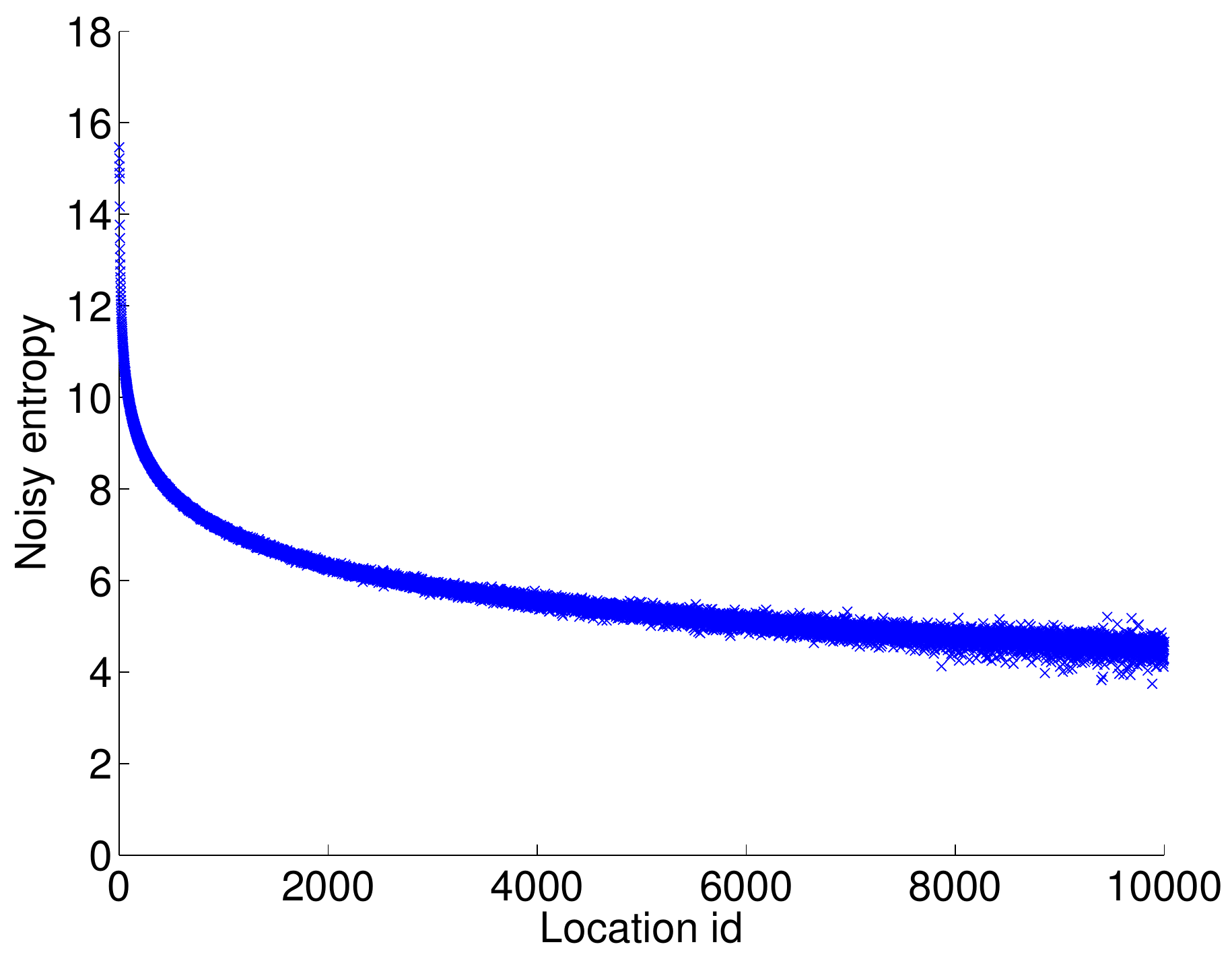}
		\subcaption{{\LMS} (Dense)}
		\label{fig:d-limit-ss}
	\end{minipage}
	\begin{minipage}[b]{.245\linewidth}
		\centering
		\includegraphics[width=1\textwidth]{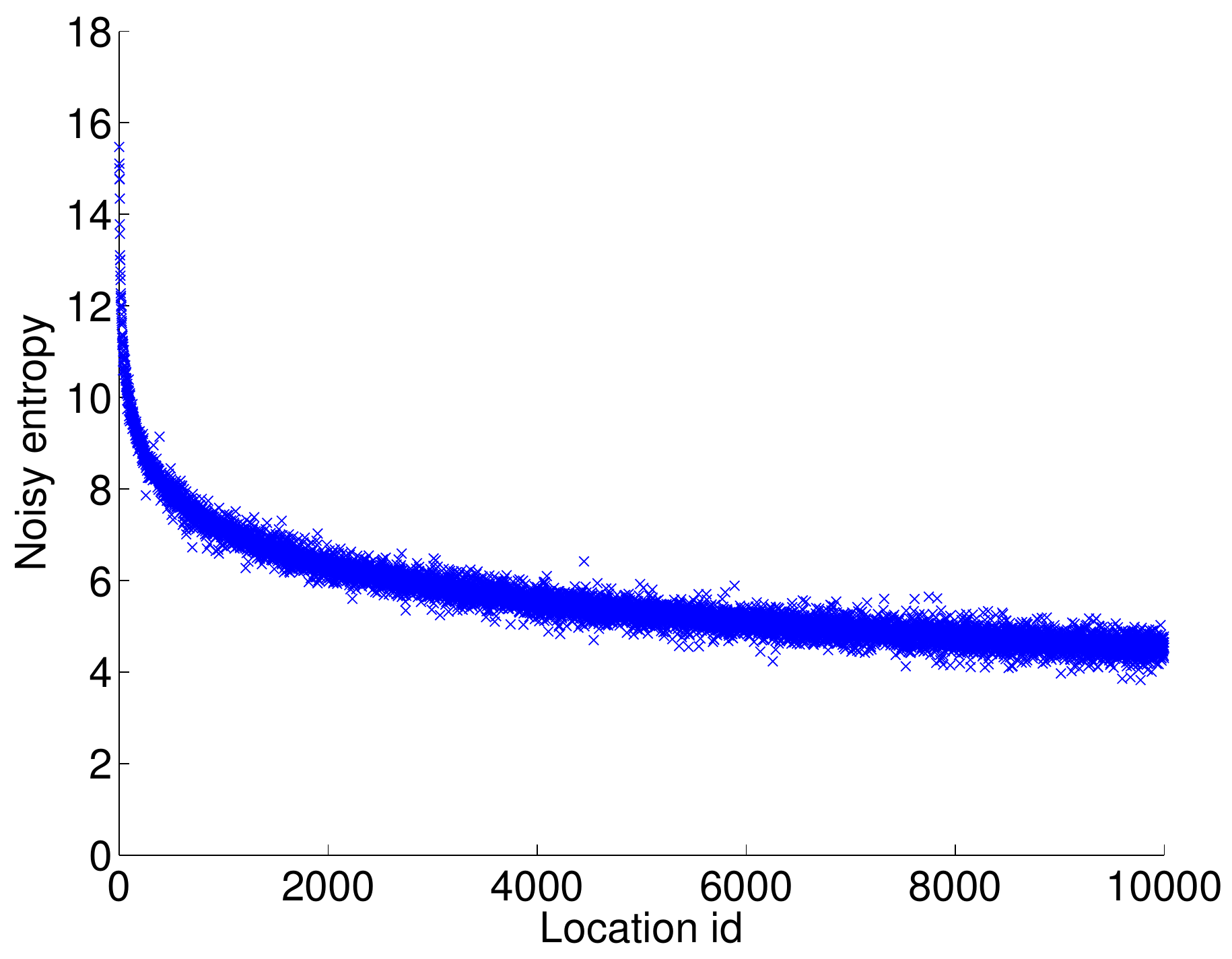}
		\subcaption{{\LMC} (Dense)}
		\label{fig:d-limit-cb}
	\end{minipage}

	\begin{minipage}[b]{.245\linewidth}
		\centering
		\includegraphics[width=1\textwidth]{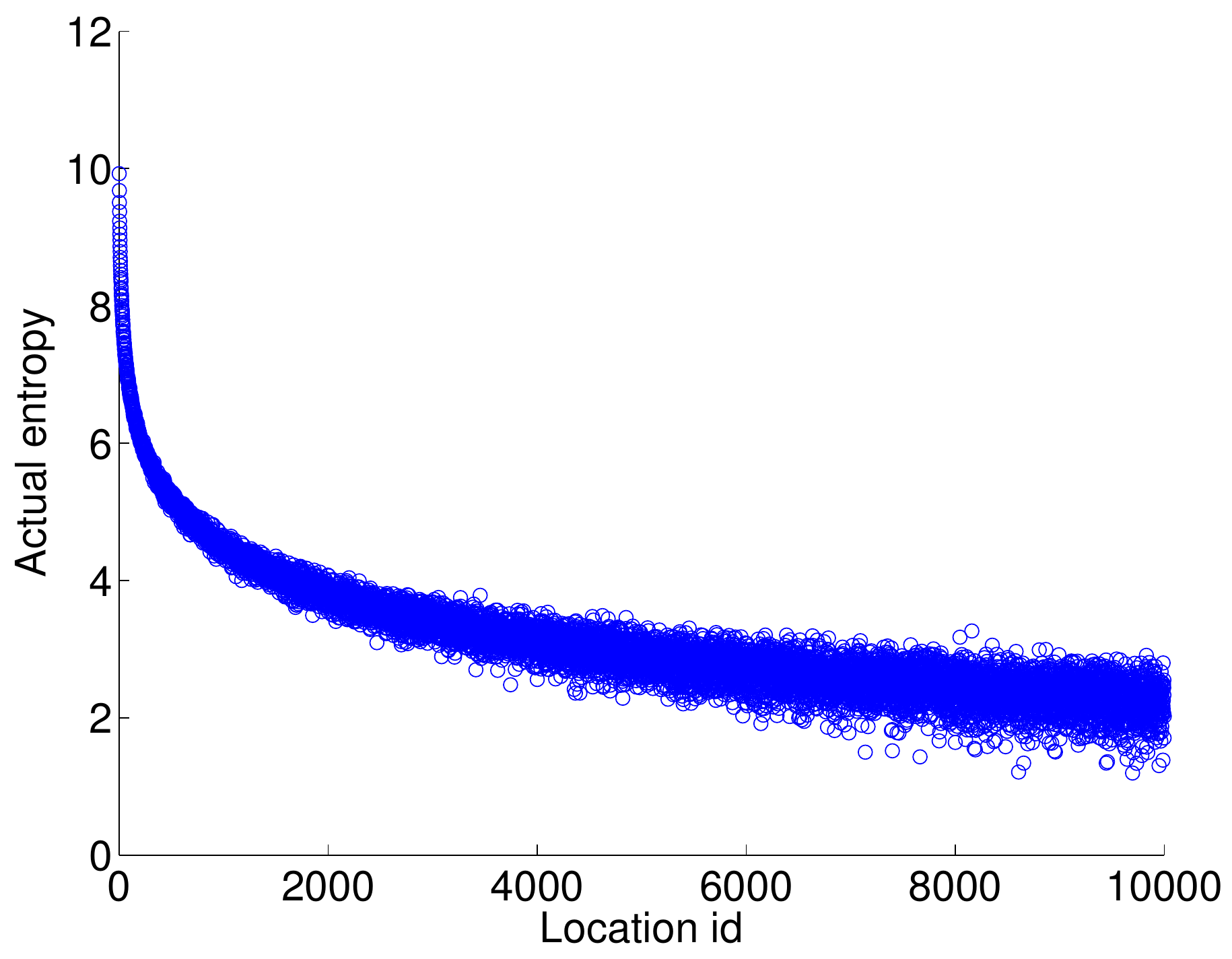}
		\subcaption{Actual (Sparse)}
		\label{fig:s-actual}
	\end{minipage}
	\begin{minipage}[b]{.245\linewidth}
		\centering
		\includegraphics[width=1\textwidth]{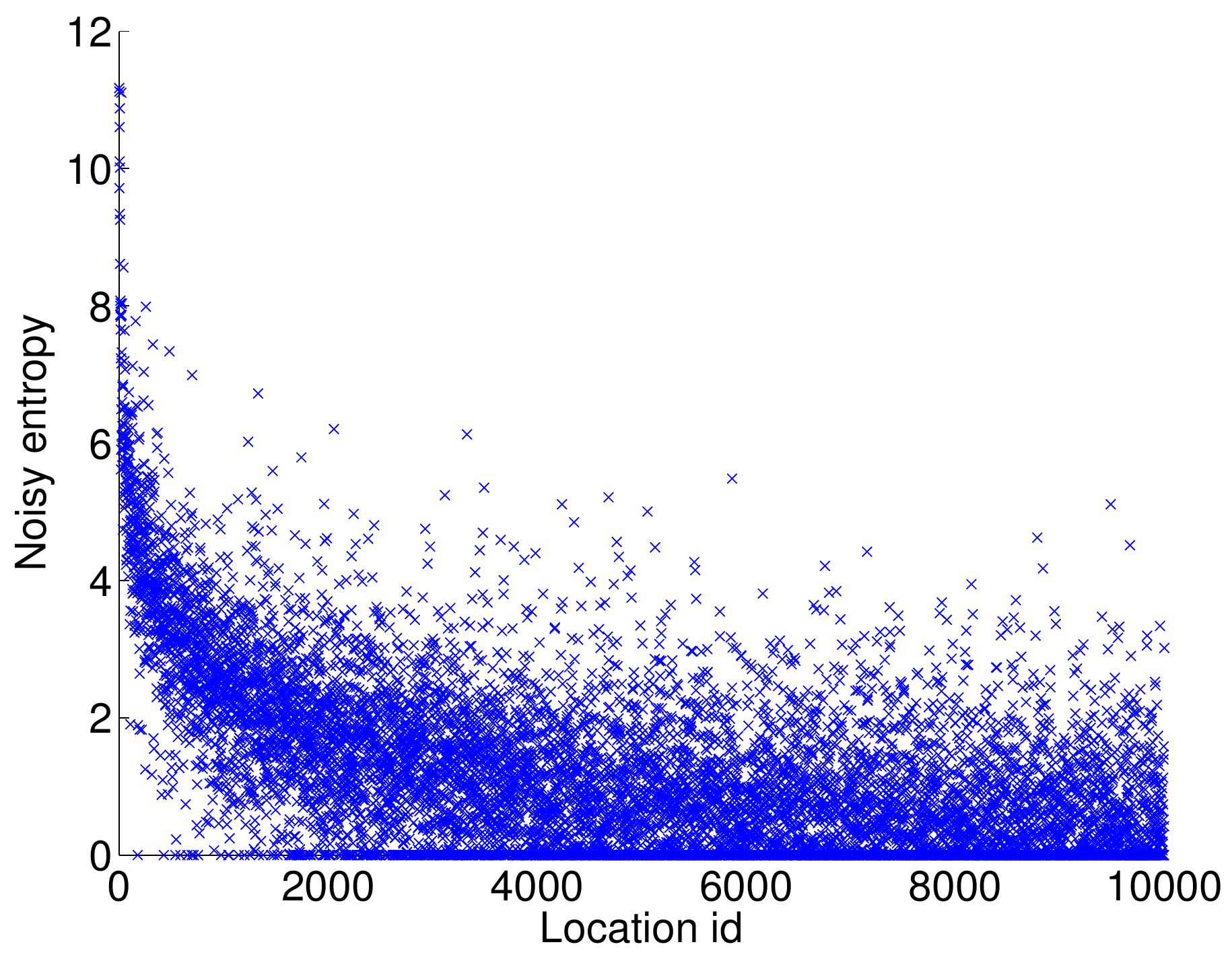}
		\subcaption{{\LM} (Sparse)}
		\label{fig:s-limit}
	\end{minipage}
	\begin{minipage}[b]{.245\linewidth}
		\centering
		\includegraphics[width=1\textwidth]{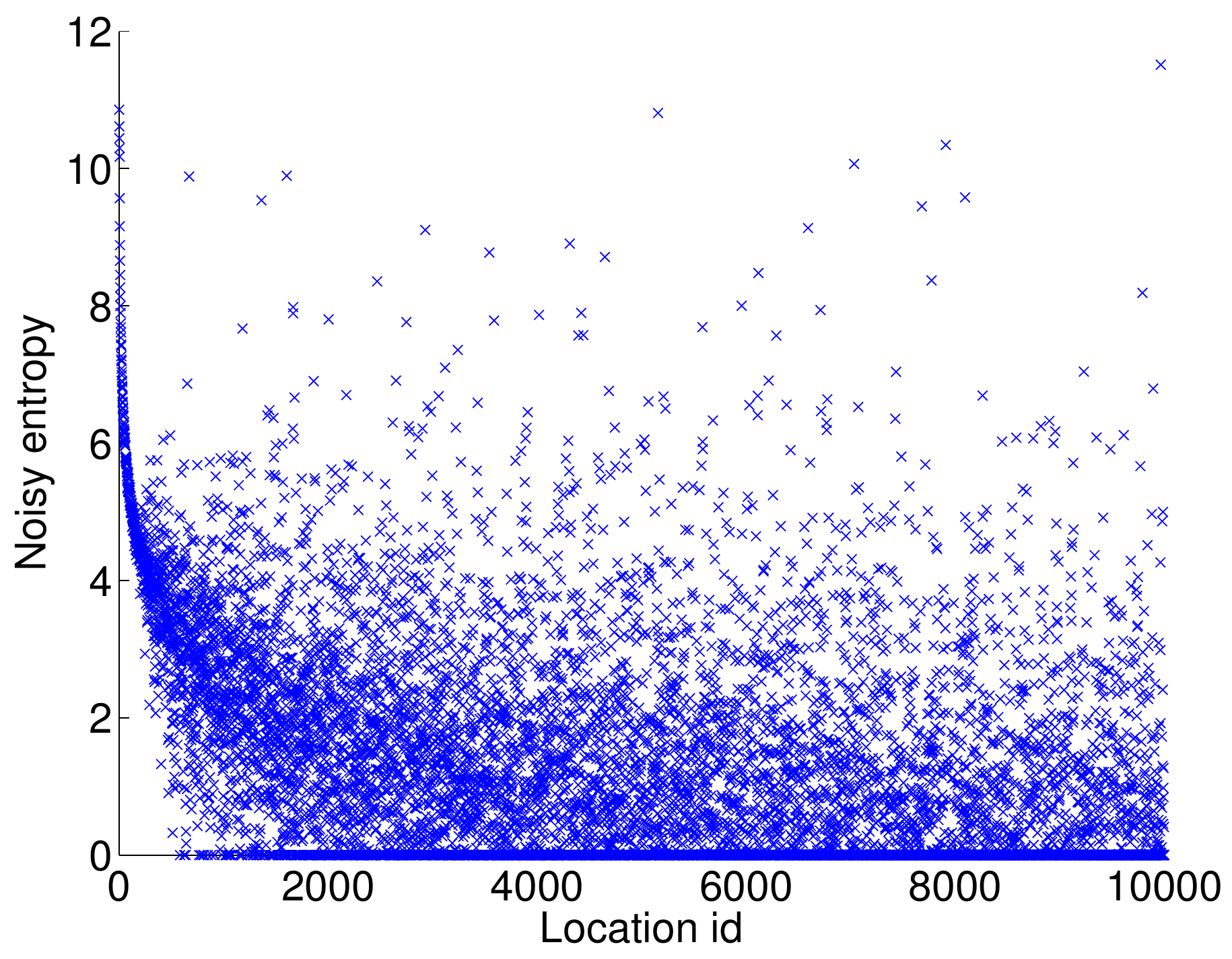}
		\subcaption{{\LMS} (Sparse)}
		\label{fig:s-limit-ss}
	\end{minipage}
	\begin{minipage}[b]{.245\linewidth}
		\centering
		\includegraphics[width=1\textwidth]{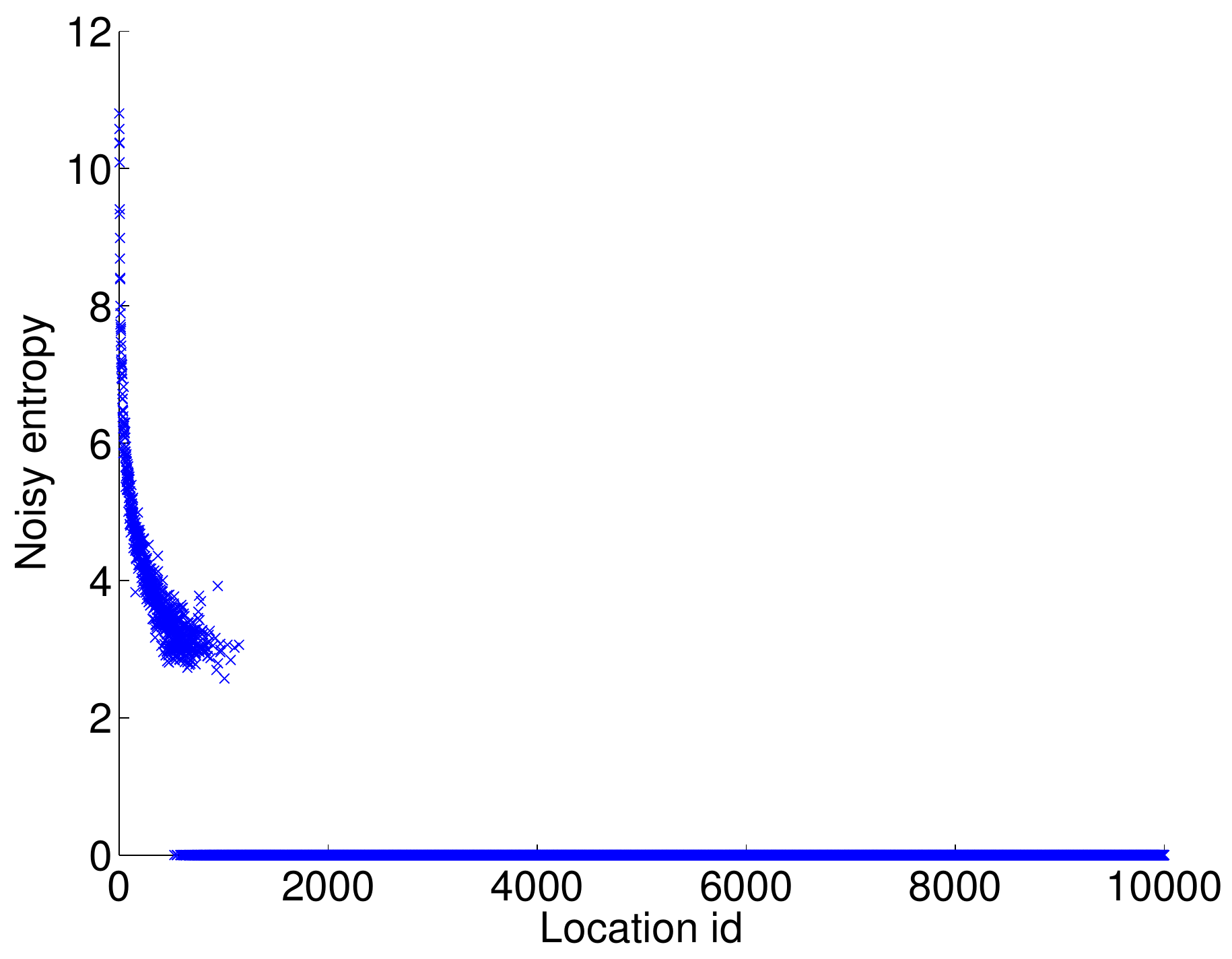}
		\subcaption{{\LMC} (Sparse)}
		\label{fig:s-limit-cb}
	\end{minipage}
	\caption{Comparison of the distributions of noisy vs. actual location entropy on the dense and sparse datasets.}
	\label{fig:raw}
\end{figure*}

Due to the truncation technique, some locations may be discarded. Thus, we report the percentage of perturbed locations, named {\em published ratio}. The published ratio is computed as the number of perturbed locations divided by the total number of \emph{eligible} locations. A location is eligible for publication if it contains check-ins from at least $K$ users ($K\ge 1$).
Figure~\ref{fig:published_ratio} shows the effect of $k$ on the published ratio of {\LMC}. Note that the published ratio of {\LM} and {\LMS} is the same as {\LMC} when $k=K$. The figure shows that the ratio is 100\% with Dense, while that of Sparse is less than 10\%. The reason is that with Dense, each location is visited by a large number of users on average (see Table~\ref{tab:datasets}); thus, limiting $M$ and $C$ would reduce the average number of users visiting a location but not to the point where the locations are suppressed. This result suggests that our truncation technique performs well on large datasets.

\begin{figure}[!htb]\centering
  \includegraphics[width=0.32\textwidth]{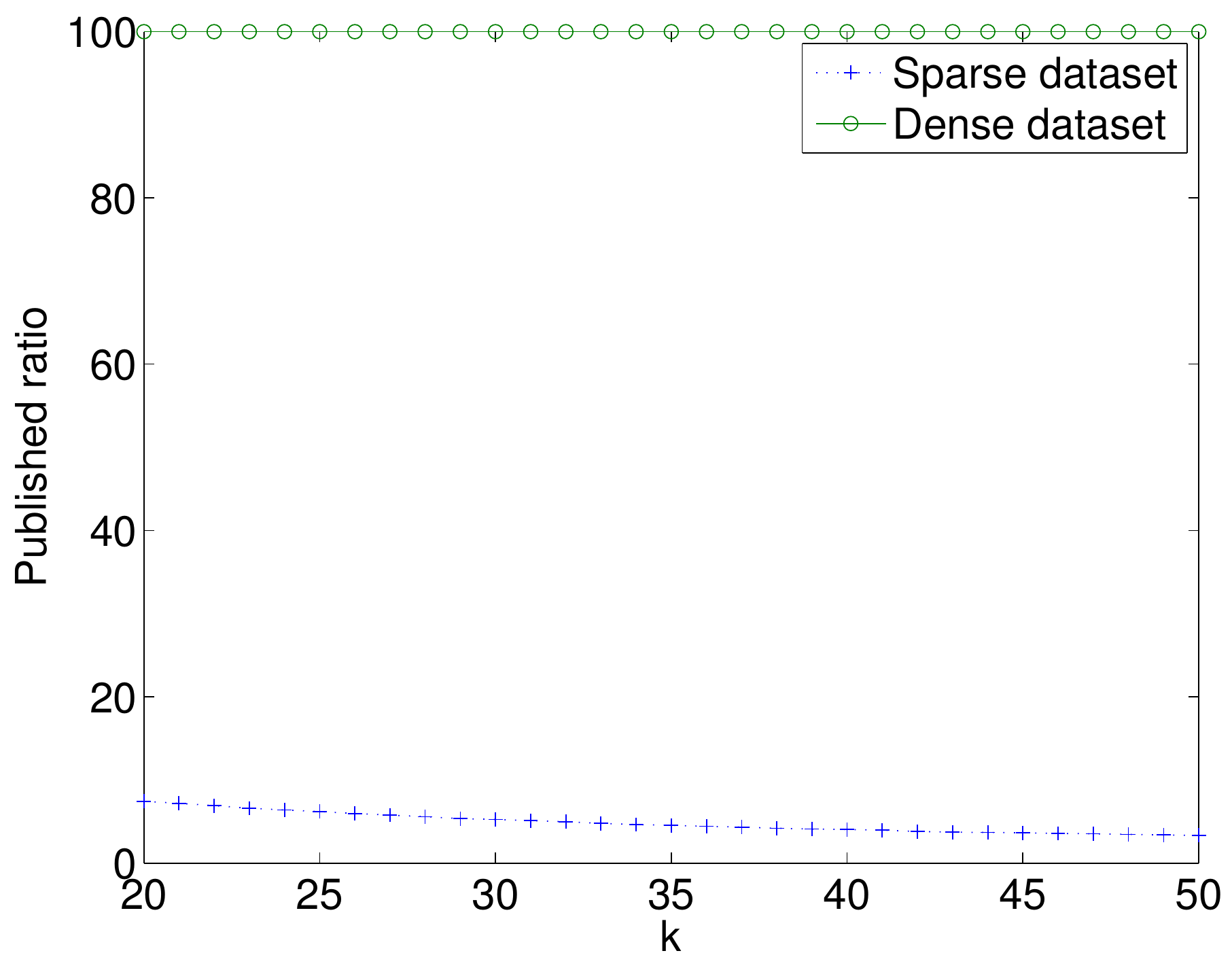}
  \caption{Published ratio of {\LMC} when varying $k$ ($K=20$).}
  \label{fig:published_ratio}
\end{figure}

\subsubsection{Privacy-Utility Trade-off (Varying $\epsilon$)}
\label{sec:vary_eps}

We compare the trade-off between privacy and utility by varying the privacy budget $\epsilon$. The utility is captured by the KL-divergence metric introduced in Section~\ref{sec:setup}. We also use the MSE metric.
Figure~\ref{fig:vary-e} illustrates the results.
As expected, when $\epsilon$ increases, less noise is injected, and values of KL-divergence and MSE decrease. Interestingly though, KL-divergence and MSE saturate at $\epsilon=5$, where reducing privacy level (increase $\epsilon$) only marginally increases utility. This can be explained through a significant approximation error in our thresholding technique that outweighs the impact of having smaller perturbation error.
Note that the approximation errors are constant in this set of experiments since the parameters $C$, $M$ and $k$ are fixed.

Another observation is that the observed errors incurred are generally higher for Dense (Figures~\ref{fig:d-mse-vary-e} vs.~\ref{fig:s-mse-vary-e}), which is surprising, as differentially private algorithms often perform better on dense datasets. The reason for this is because  limiting $M$ and $C$ has a larger impact on Dense, resulting in a large perturbation error.
Furthermore, we observe that the improvements of {\LMS} and {\LMC} over {\LM} are more significant with small $\epsilon$. In other words, {\LMS} and {\LMC} would have more impact with a higher level of privacy protection. Note that these enhancements come at the cost of weaker privacy protection.

\begin{figure}[ht]
	\begin{minipage}[b]{.49\linewidth}
		\centering
		\includegraphics[width=1\textwidth]{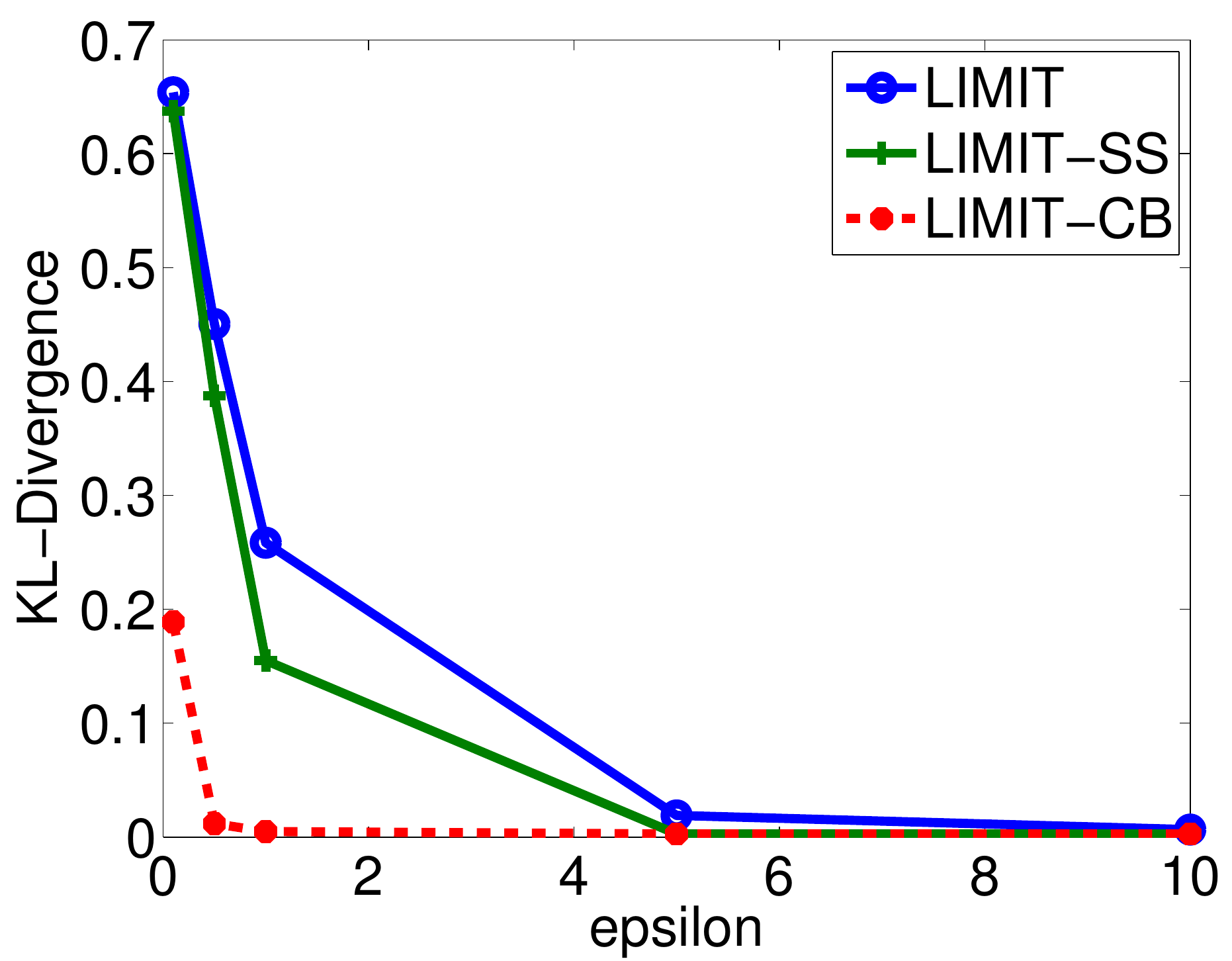}
		\subcaption{Dense}
		\label{fig:d-kl-vary-e}
	\end{minipage}
	\begin{minipage}[b]{.49\linewidth}
		\centering
		\includegraphics[width=1\textwidth]{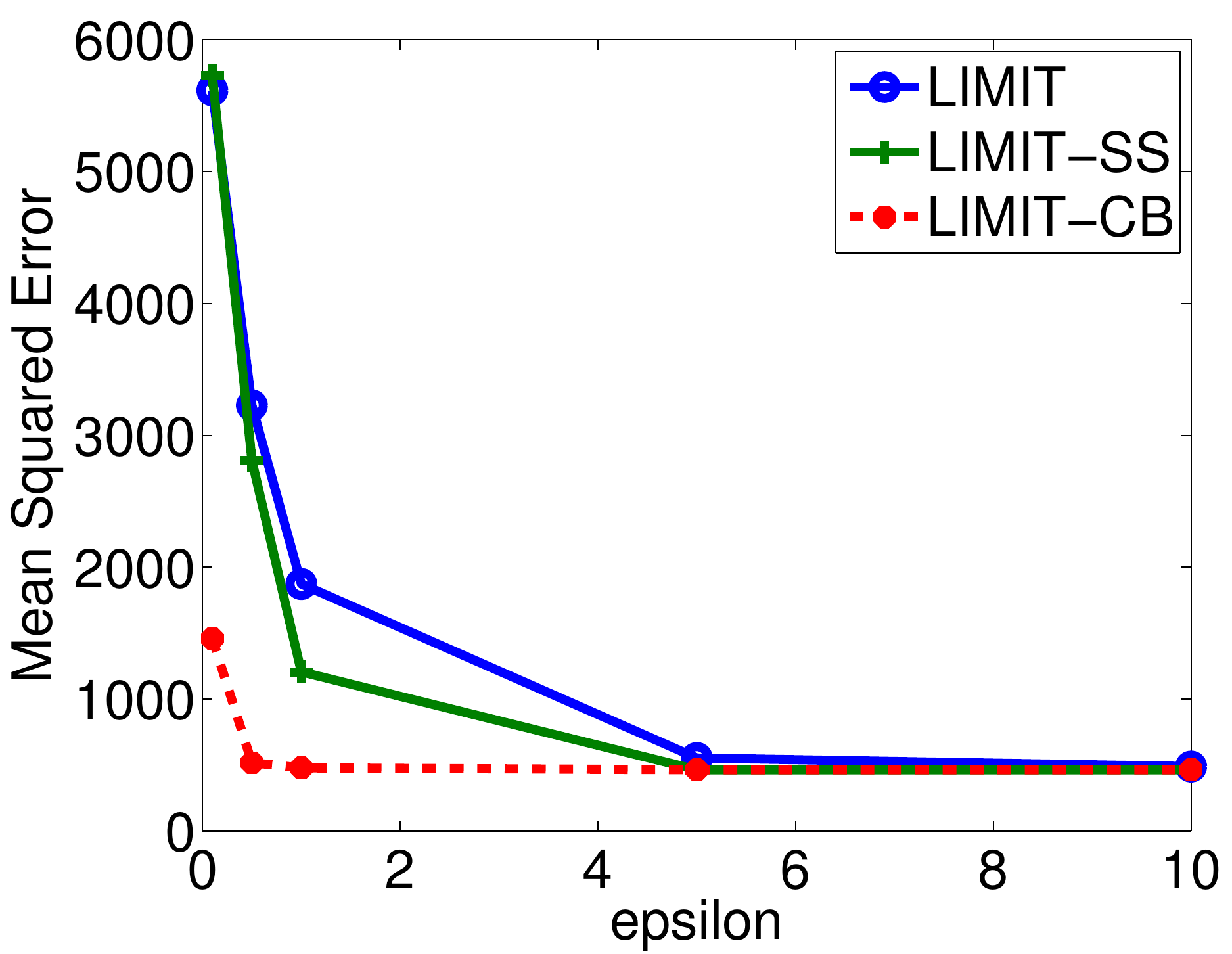}
		\subcaption{Dense}
		\label{fig:d-mse-vary-e}
	\end{minipage}
	\begin{minipage}[b]{.49\linewidth}
		\centering
		\includegraphics[width=1\textwidth]{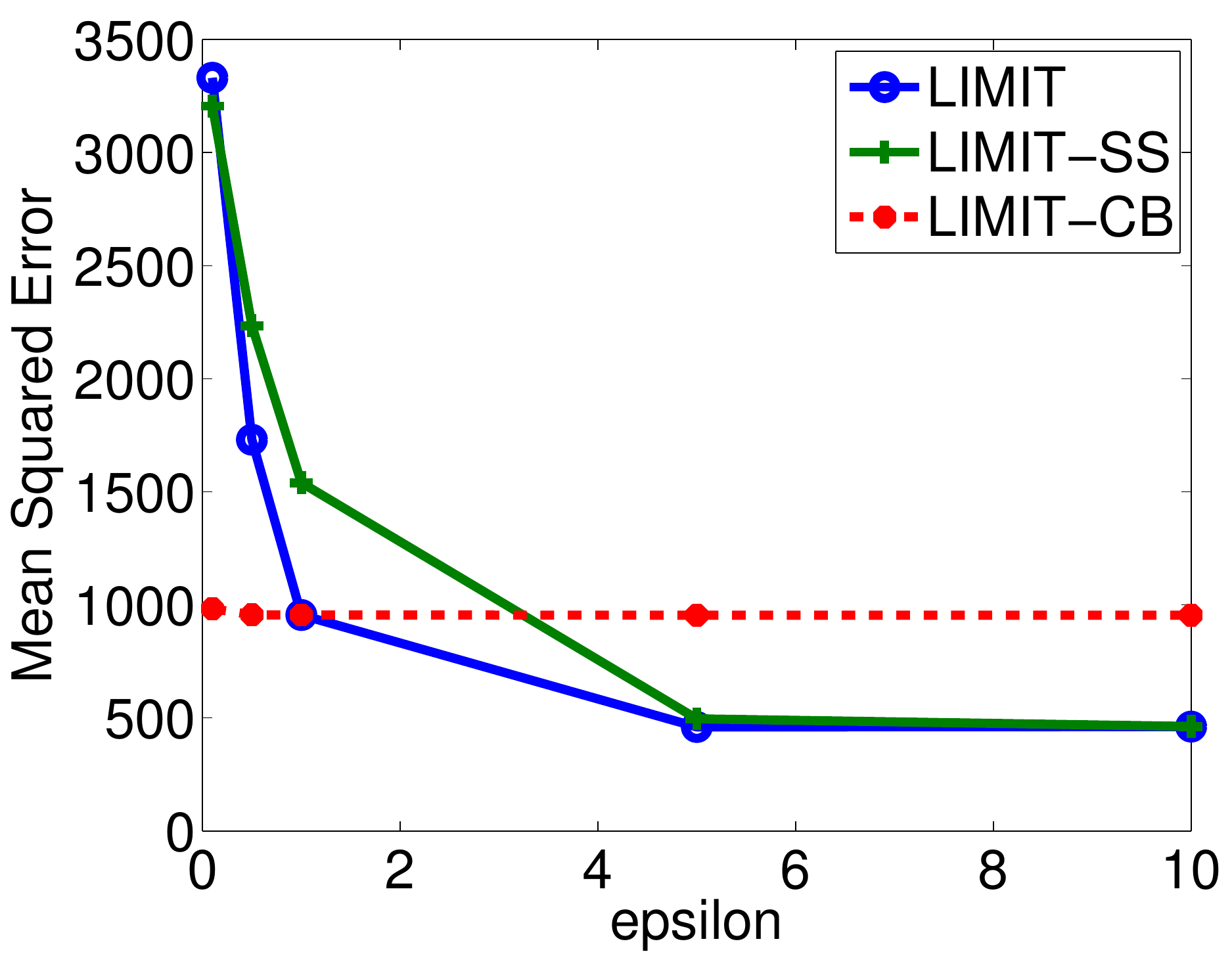}
		\subcaption{Sparse}
		\label{fig:s-mse-vary-e}
	\end{minipage}
	\begin{minipage}[b]{.49\linewidth}
		\centering
		\includegraphics[width=1\textwidth]{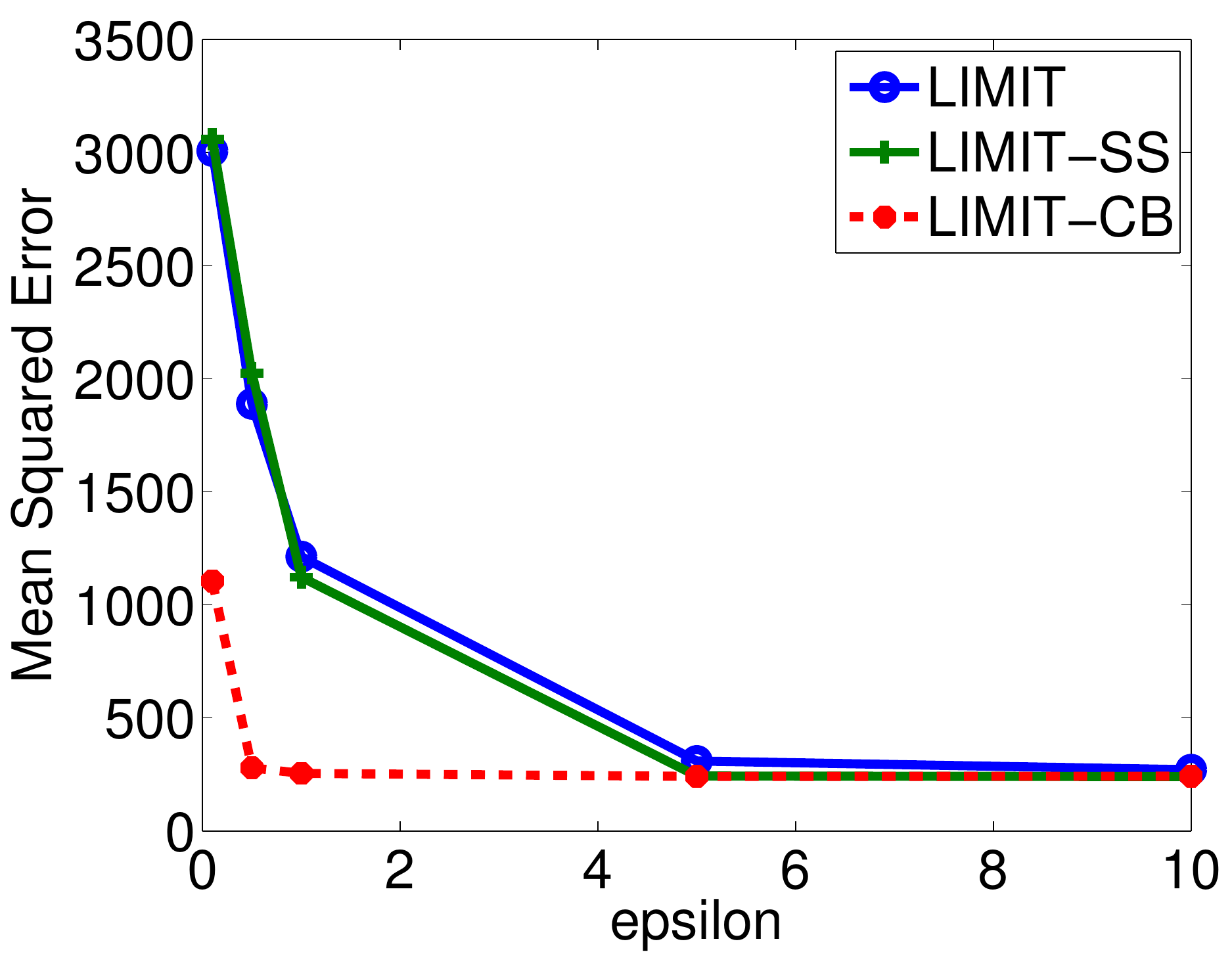}
		\subcaption{Sparse, Throwaway}
		\label{fig:s-mse-vary-e-p}
	\end{minipage}
	\caption{Varying $\epsilon$}
	\label{fig:vary-e}
\end{figure}

\subsubsection{The Effect of Varying $M$ and $C$}

We first evaluate the performance of our proposed techniques by varying threshold $M$.
For brevity, we present the results only for MSE, as similar results have been observed for KL-divergence.
Figure~\ref{fig:mse-vary-m} indicates the trade-off between the approximation error and the perturbation error. Our thresholding technique decreases $M$ to reduce the perturbation error, but at the cost of increasing the approximation error. As a result, at a particular value of $M$, the technique balances the two types of errors and thus minimizes the total error. For example, in Figure~\ref{fig:d-mse-vary-m-p}, {\LM} performs best at $M=5$, while {\LMS} and {\LMC} work best at $M\ge 30$. In Figure~\ref{fig:s-mse-vary-m-p}, however, {\LMS} performs best at $M=10$ and {\LMC} performs best at $M=20$.

\begin{figure}[ht]
	\begin{minipage}[b]{.49\linewidth}
		\centering
		\includegraphics[width=1\textwidth]{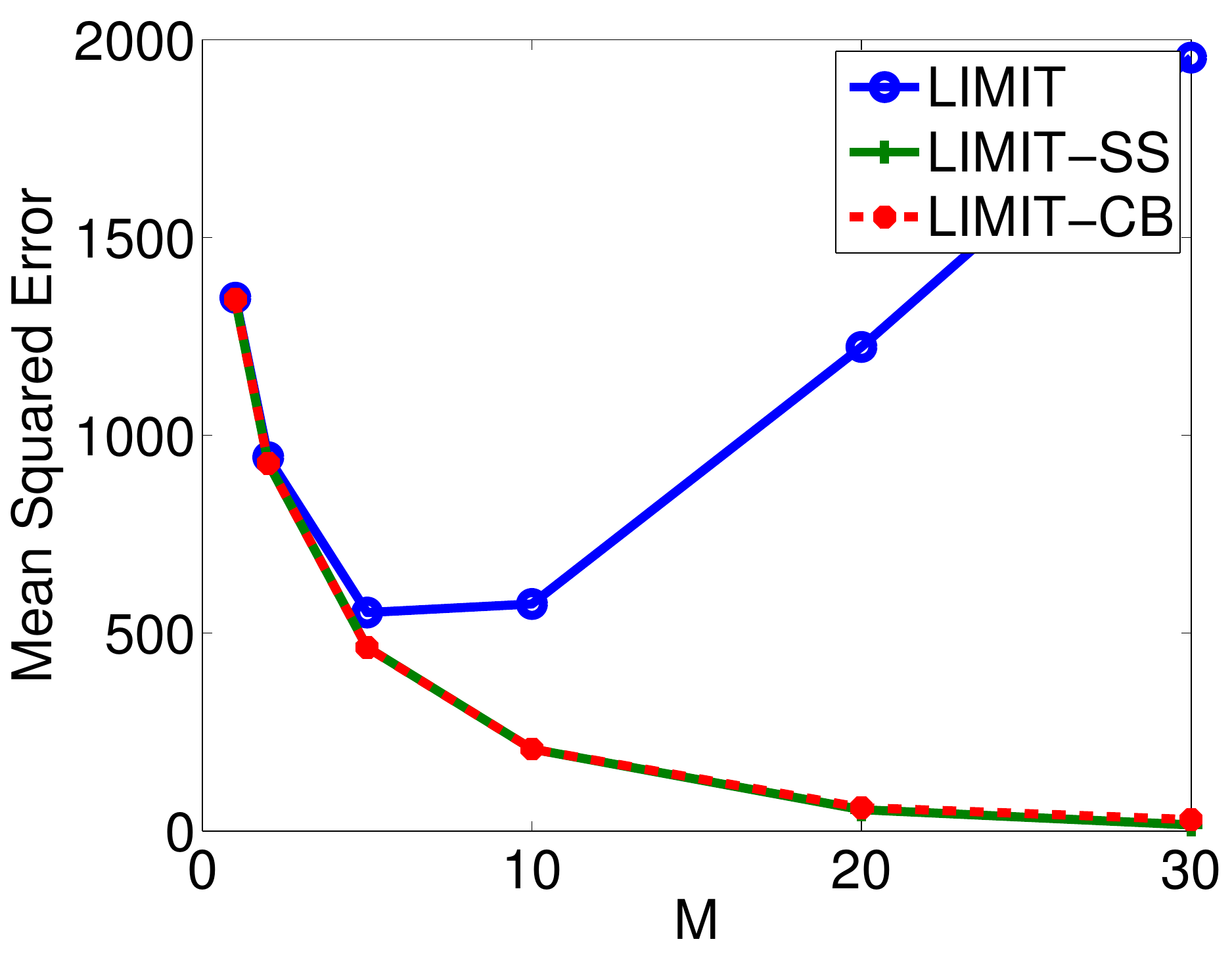}
		\subcaption{Dense}
		\label{fig:d-mse-vary-m-p}
	\end{minipage}
	\begin{minipage}[b]{.49\linewidth}
		\centering
		\includegraphics[width=1\textwidth]{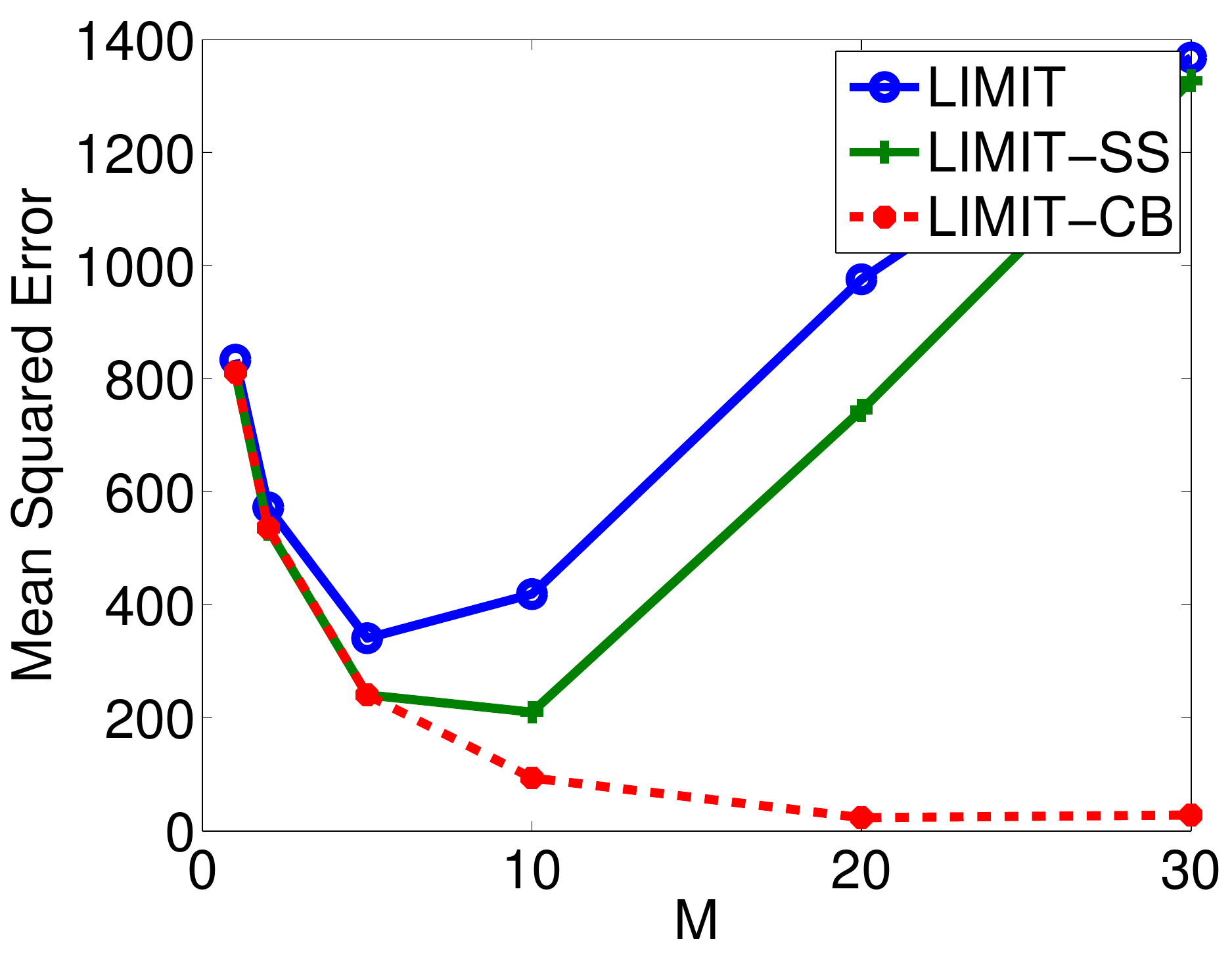}
		\subcaption{Sparse, Throwaway}
		\label{fig:s-mse-vary-m-p}
	\end{minipage}
	\caption{Varying $M$}
	\label{fig:mse-vary-m}
\end{figure}

We then evaluate the performance of our techniques by varying threshold $C$. Figure~\ref{fig:kl-vary-c} shows the results. For brevity, we only include KL-divergence results (MSE metric shows similar trends).
The graphs show that KL-divergence increases as $C$ grows. This observation suggests that $C$ should be set to a small number (less than 10). By comparing the effect of varying $M$ and $C$, we conclude that $M$ has more impact on the trade-off between the approximation error and the perturbation error.
\begin{figure}[ht]
	\begin{minipage}[b]{.49\linewidth}
		\centering
		\includegraphics[width=1\textwidth]{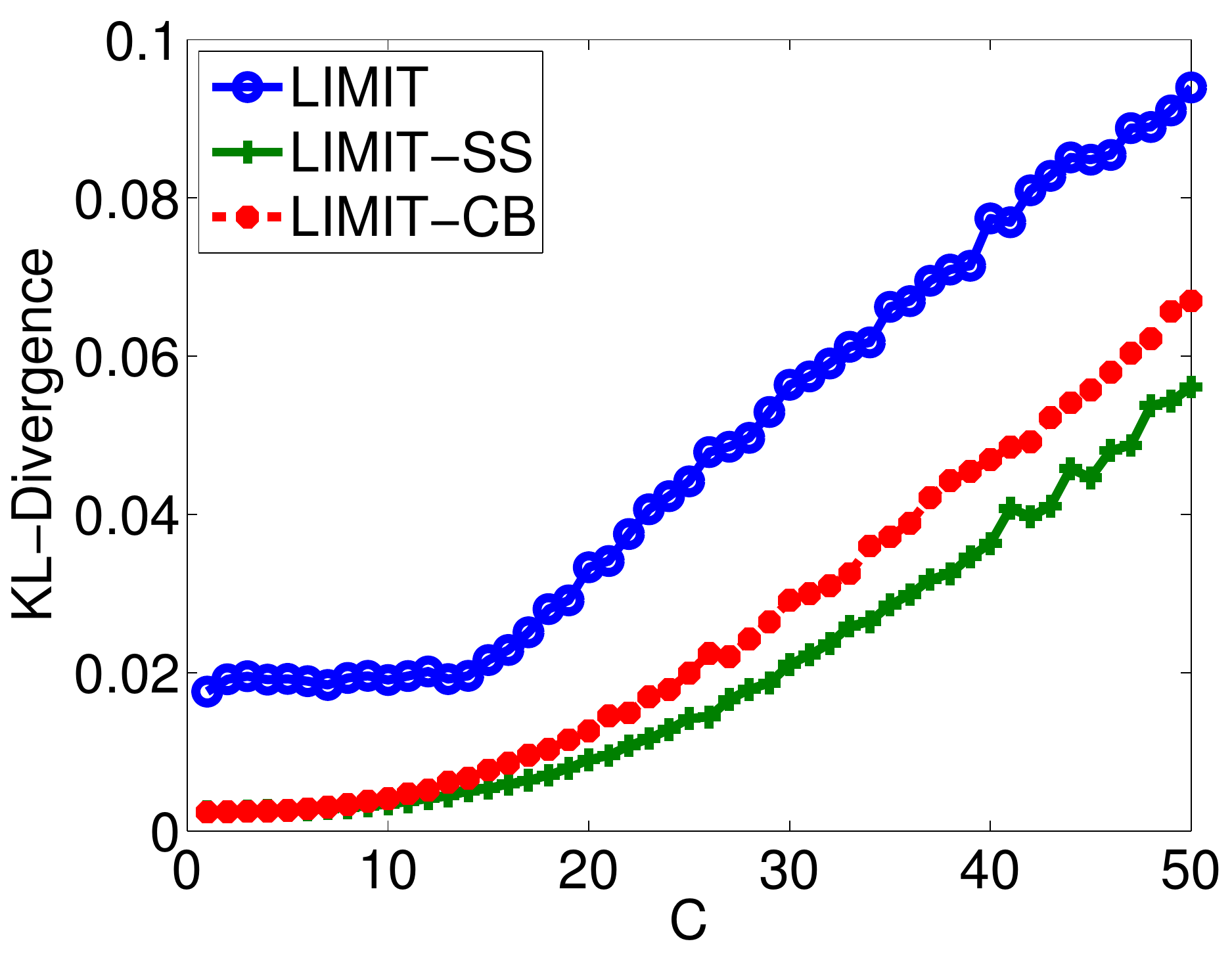}
		\subcaption{Dense}
		\label{fig:d-kl-vary-c-p}
	\end{minipage}
	\begin{minipage}[b]{.49\linewidth}
		\centering
		\includegraphics[width=1\textwidth]{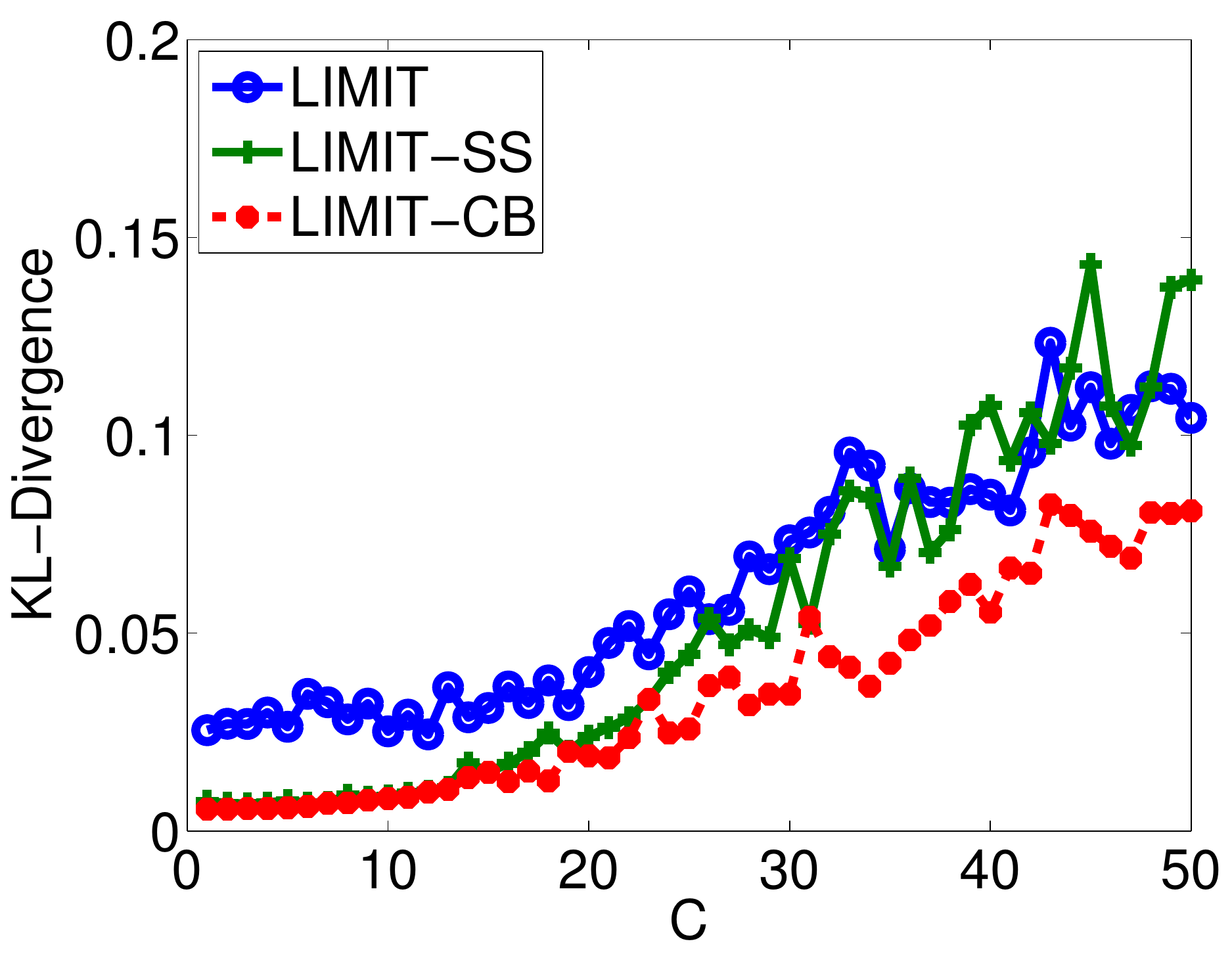}
		\subcaption{Sparse}
		\label{fig:s-kl-vary-c-p}
	\end{minipage}
	\caption{Varying $C$}
	\label{fig:kl-vary-c}
\end{figure}

\subsubsection{Results on the Gowalla Dataset}

In this section we evaluate the performance of our algorithms on the Gowalla dataset.
Figure~\ref{fig:go-raw} shows the distributions of noisy vs. actual location entropy. Note that we sort the locations based on their actual values of LE as depicted in Figure~\ref{fig:go-actual}. As expected, due to the sparseness of Gowalla (see Table~\ref{tab:datasets}),  the published values of LE in {\LM} and {\LMS} are scattered while those in {\LMC} preserve the trend in the actual data but at the cost of throwing away more locations (Figure~\ref{fig:go-limit-cb}). Furthermore, we conduct experiments on varying various parameters (i.e., $\epsilon,C,M,k$) and observe trends similar to the Sparse dataset; nevertheless, for brevity, we only show the impact of varying $\epsilon$ and $M$ in Figure~\ref{fig:go-mse}.

\begin{figure}[ht]
	\begin{minipage}[b]{.49\linewidth}
		\centering
		\includegraphics[width=1\textwidth]{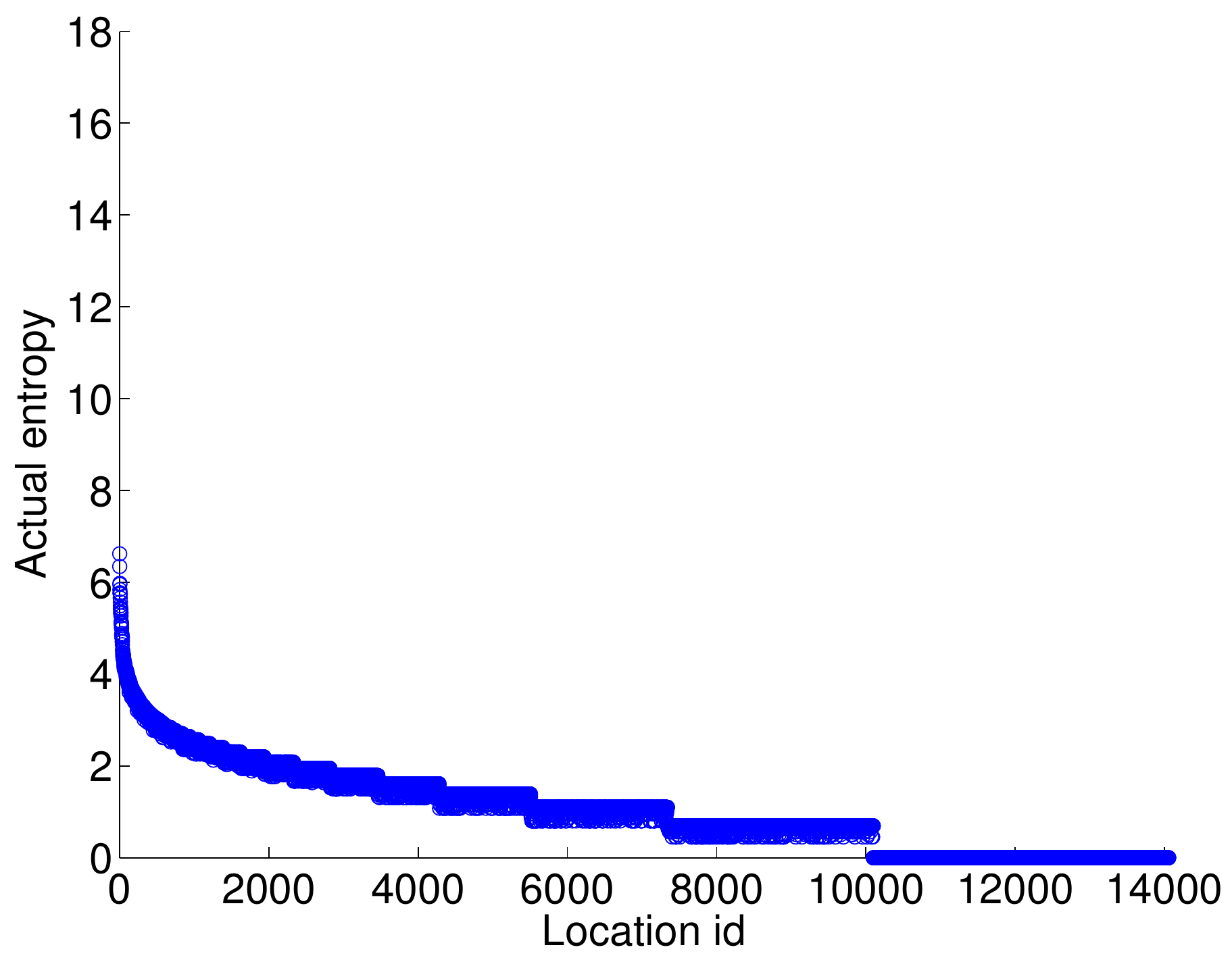}
		\subcaption{Actual}
		\label{fig:go-actual}
	\end{minipage}
	\begin{minipage}[b]{.49\linewidth}
		\centering
		\includegraphics[width=1\textwidth]{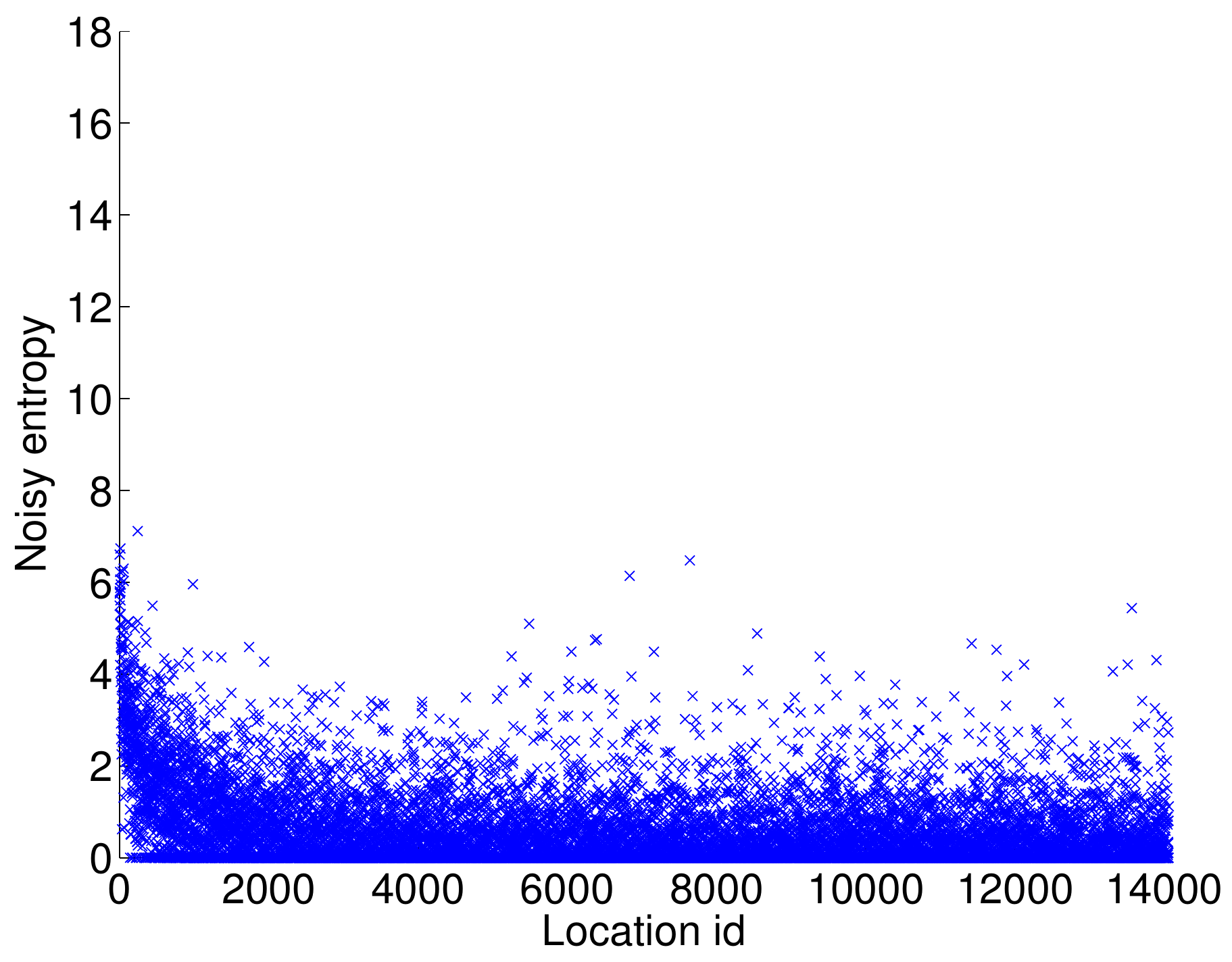}
		\subcaption{{\LM}}
		\label{fig:go-limit}
	\end{minipage}
	\begin{minipage}[b]{.49\linewidth}
		\centering
		\includegraphics[width=1\textwidth]{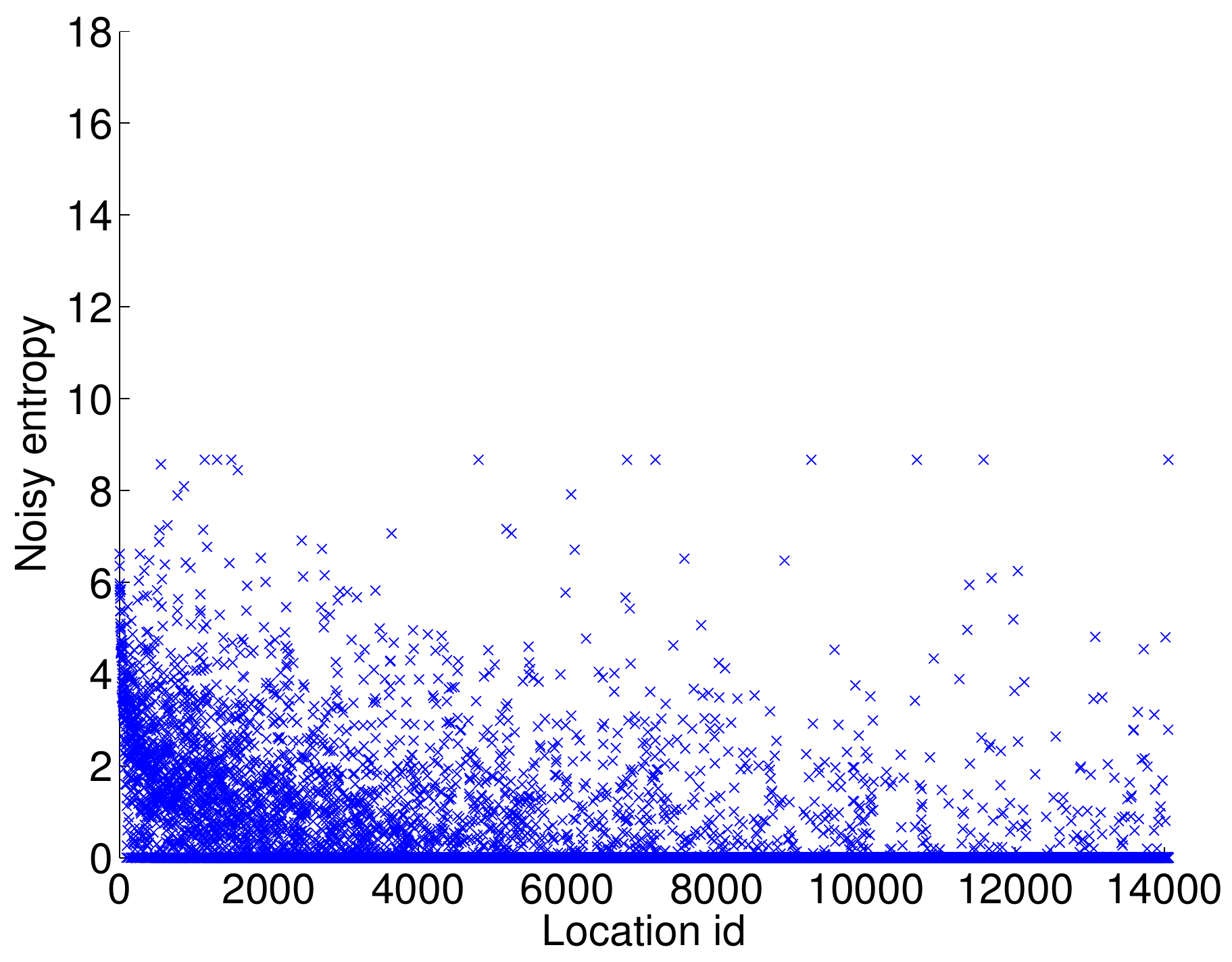}
		\subcaption{{\LMS}}
		\label{fig:go-limit-ss}
	\end{minipage}
	\begin{minipage}[b]{.49\linewidth}
		\centering
		\includegraphics[width=1\textwidth]{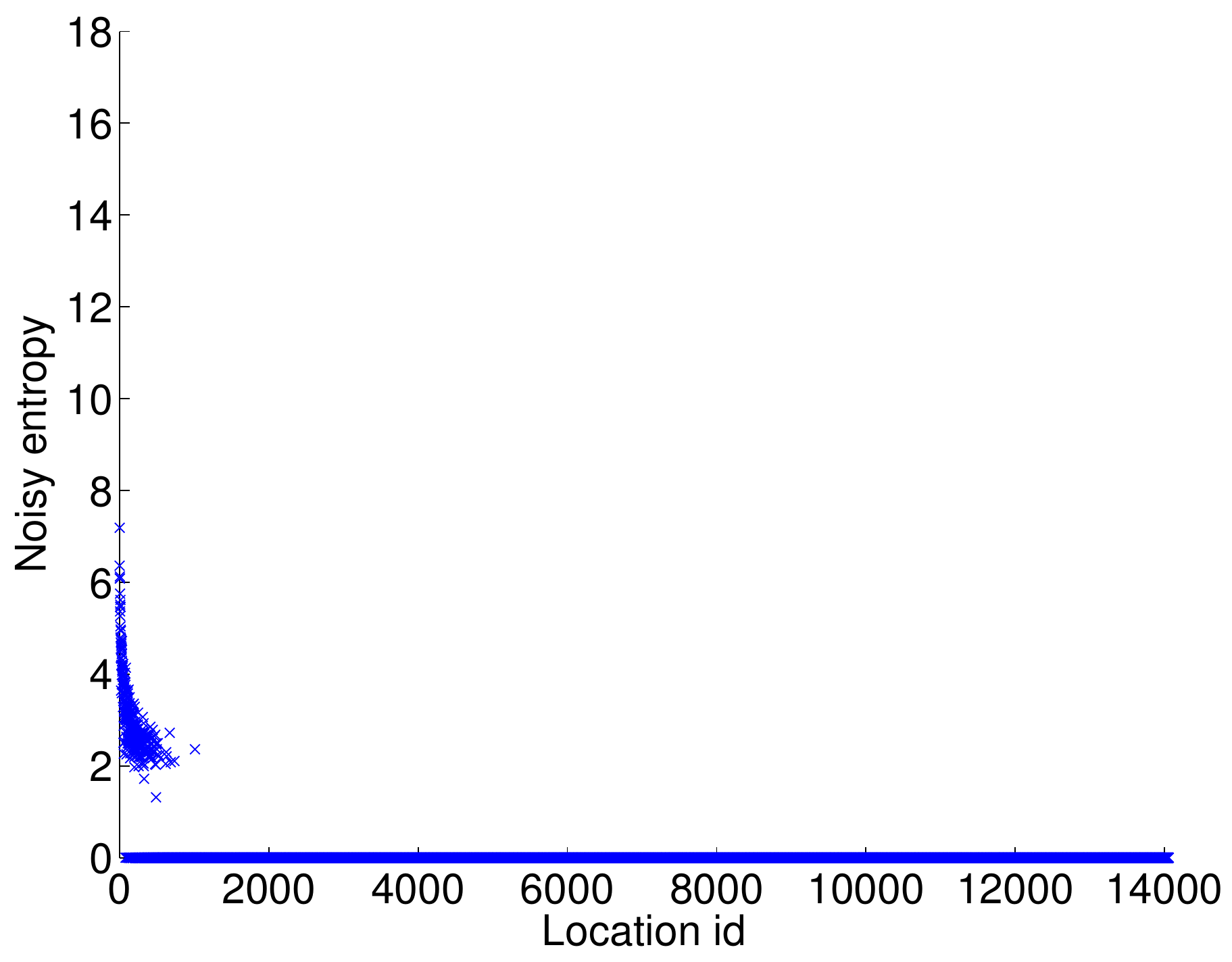}
		\subcaption{{\LMC}}
		\label{fig:go-limit-cb}
	\end{minipage}
	\caption{Comparison of the distributions of noisy vs. actual location entropy on Gowalla, $M=5$.}
	\label{fig:go-raw}
\end{figure}

\begin{figure}[ht]
	\begin{minipage}[b]{.49\linewidth}
		\centering
		\includegraphics[width=1\textwidth]{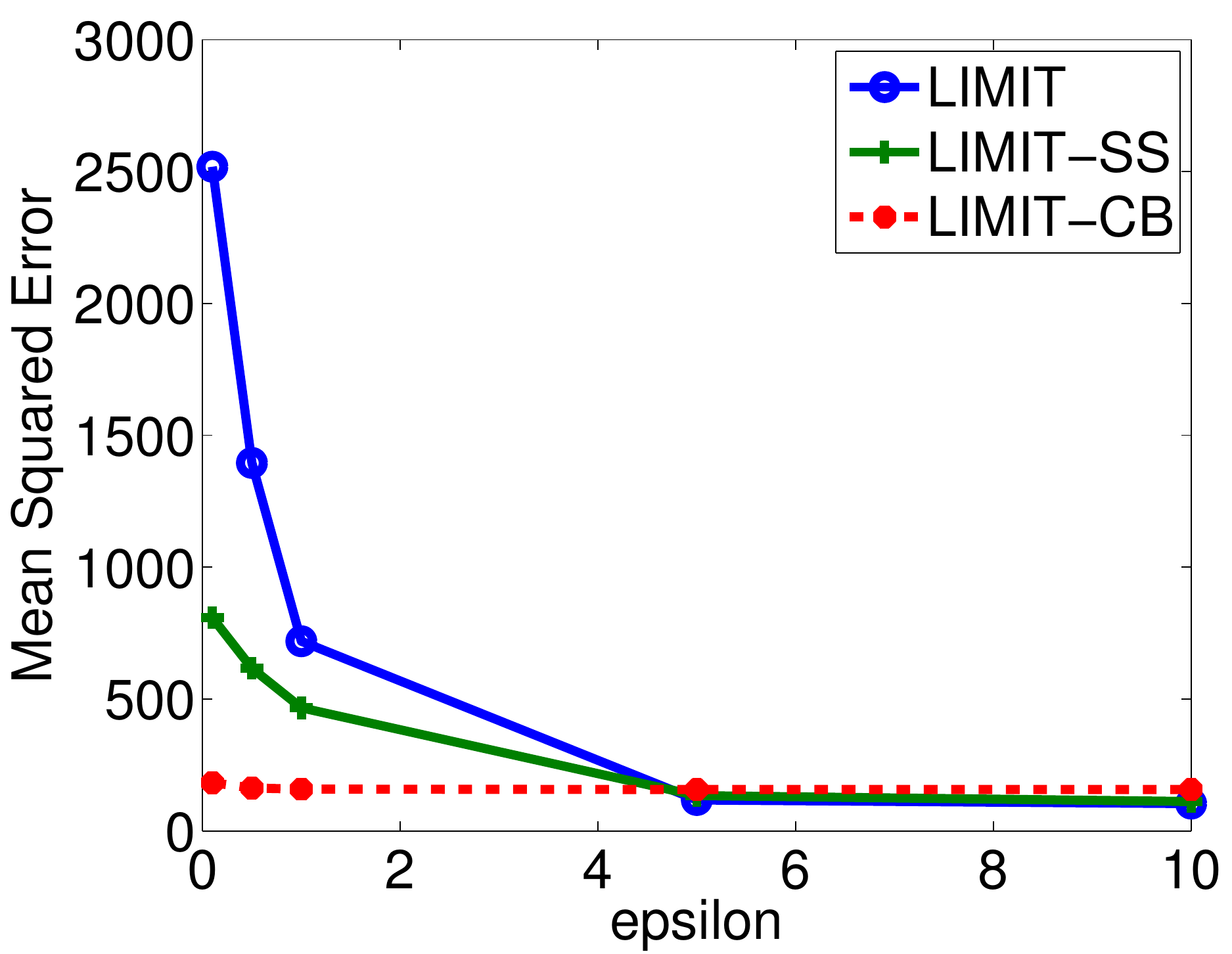}
		\subcaption{Vary $\epsilon$}
		\label{fig:go-mse-vary-e}
	\end{minipage}
	\begin{minipage}[b]{.49\linewidth}
		\centering
		\includegraphics[width=1\textwidth]{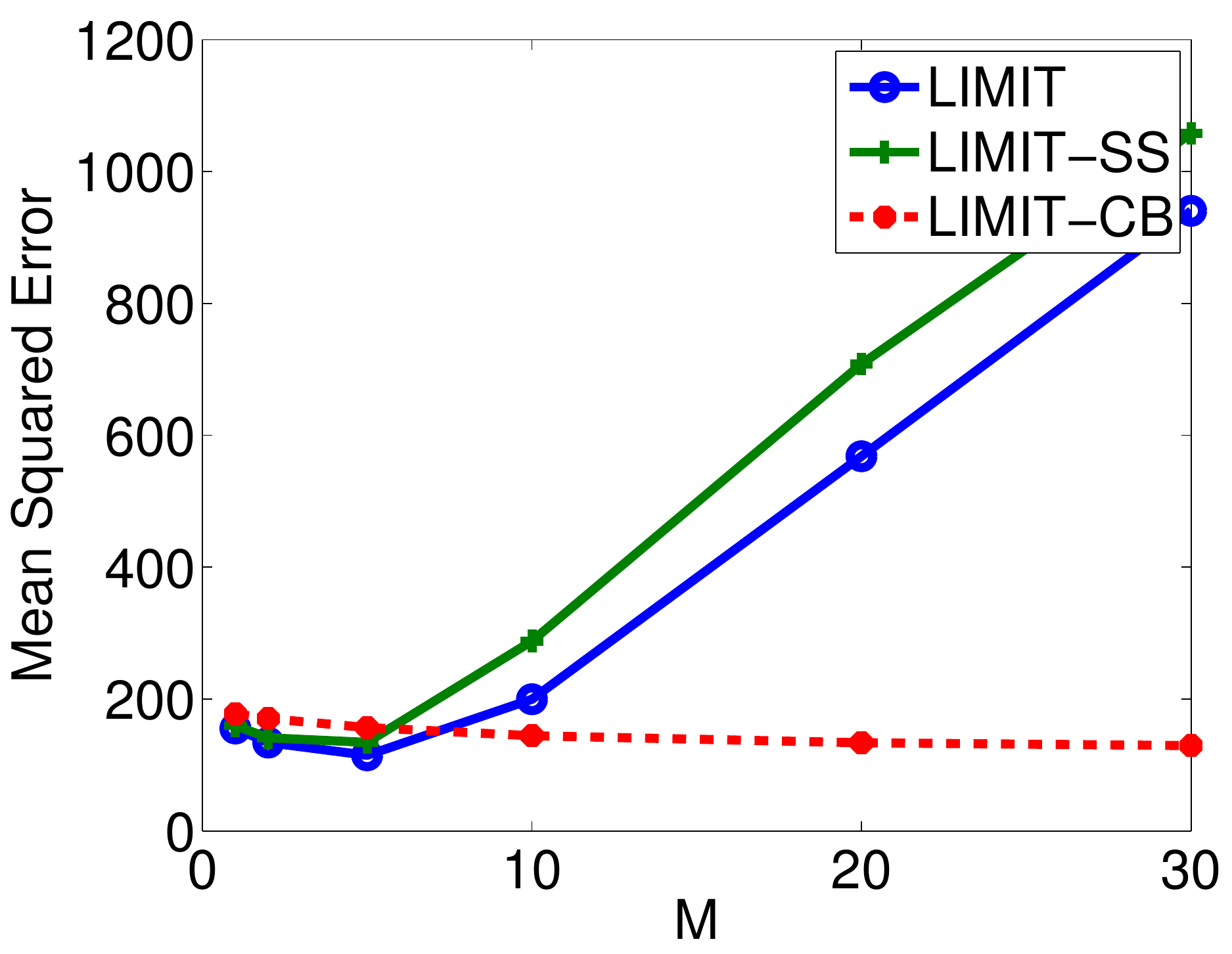}
		\subcaption{Vary $M$}
		\label{fig:go-mse-vary-m}
	\end{minipage}
	\caption{Varying $\epsilon$ and $M$ (Gowalla).}
	\label{fig:go-mse}
\end{figure}

\textbf{Recommendations for Data Releases:}
We summarize our observations and provide guidelines for choosing appropriate techniques and parameters.
{\LMC} generally performs best on sparse datasets because it only focuses on publishing the locations with large visits. Alternatively, if the dataset is dense, {\LMS} is recommended over {\LMC} since there are sufficient locations with large visits.
A dataset is dense if most locations (e.g., 90\%) have at least $n_{CB}$ users, where $n_{CB}$ is the threshold for choosing {\LMC}.  
Particularly, given fixed parameters $C, \epsilon, \delta, k$---$n_{CB}$ can be found by comparing the global sensitivity of {\LMC} and the precomputed smooth sensitivity.
In Figure~\ref{fig:sensitivity_m5}, $n_{CB}$ is a particular value of $n$ where $SS(C, n_{CB})$ is smaller than the global sensitivity of {\LMC}. In other words, the noise magnitude required for {\LMS} is smaller than that for {\LMC}.
Regarding the choice of parameters, to guarantee strong privacy protection, $\epsilon$ should be as small as possible, while the measured utility metrics are practical. Finally, the value of $C$ should be small ($ \le 10$), while the value of $M$ can be tuned to achieve maximum utility.

\section{Related Work}
\label{sec:related}

\textbf{Location privacy} has largely been studied in the context of location-based services, participatory sensing and spatial crowdsourcing. Most studies use the model of spatial $k$-anonymity~\cite{sweeney2002k}, where the location of a user is hidden among $k$ other users~\cite{gruteser2003anonymous,mokbel2006new}. However, there are known attacks on $k$-anonymity, e.g., when all $k$ users are at the same location.
Nevertheless, such techniques assume a centralized architecture with a trusted third party, which is a single point of attack.
Consequently, a technique that makes use of cryptographic techniques such as private information retrieval is proposed that does not rely on a trusted third party to anonymize locations~\cite{ghinita2008private}.
Recent studies on location privacy have focused on leveraging differential privacy (DP) to protect the privacy of users~\cite{to2014framework,xiao2015protecting}.

\textbf{Location entropy} has been extensively used in various areas of research, including multi-agent systems~\cite{van2001entropy}, wireless sensor networks~\cite{wang2004entropy}, geosocial networks~\cite{cranshaw2010bridging,cho2011friendship,pham2013ebm}, personalized web search~\cite{leung2010personalized}, image retrieval~\cite{yanai2009visual} and spatial crowdsourcing~\cite{kazemi2012geocrowd,to2015server,to2016real}, etc.
The study that most closely relates to ours focuses on privacy-preserving location-based services in which location entropy is used as the measure of privacy or the attacker's uncertainty~\cite{xu2009feeling,toch2010empirical}.
In~\cite{xu2009feeling}, a privacy model is proposed that discloses a location on behalf of a user only if the location has at least the same popularity (quantified by location entropy) as a public region specified by a user.
In fact, locations with high entropy are more likely to be shared (checked-in) than places with low entropy~\cite{toch2010empirical}.
However, directly using location entropy compromises the privacy of individuals. For example, an adversary certainly knows whether people visiting a location based on its entropy value, e.g., low value means a small number of people visit the location, and if they are all in a small geographical area, their privacy is compromised.
To the best of our knowledge, there is no study that uses differential privacy for publishing location entropy, despite its various applications that can be highly  instrumental  in protecting the privacy of individuals.

\section{Conclusions}
\label{sec:conclude}

We introduced the problem of publishing the entropy of a set of locations according to differential privacy. A baseline algorithm was proposed based on the derived tight bound for global sensitivity of the location entropy. We showed that the baseline solution requires an excessively high amount of noise to satisfy $\epsilon$-differential privacy, which renders the published results useless. A simple yet effective truncation technique was then proposed to reduce the sensitivity bound by two orders of magnitude, and this enabled publication of location entropy with reasonable utility. The utility was further enhanced by adopting smooth sensitivity and crowd-blending. We conducted extensive experiments and concluded that the proposed techniques are practical.

\section*{Acknowledgement}
We would like to thank Prof. Aleksandra Korolova for her constructive feedback during the course of this research.

This research has been funded in in part by NSF grants IIS-1320149 and CNS-1461963, National Cancer Institute, National Institutes of Health, Department of Health and Human Services, under Contract No. HHSN261201500003B, the USC Integrated Media Systems Center, and unrestricted cash gifts from Google, Northrop Grumman and Oracle. Any opinions, findings, and conclusions or recommendations expressed in this material are those of the authors and do not necessarily reflect the views of any of the sponsors.

\bibliographystyle{abbrv}
\begin{footnotesize}
\bibliography{sigproc-sp}

\begin{thebibliography}{10}

\bibitem{abadi2016deep}
M.~Abadi, A.~Chu, I.~Goodfellow, H.~B. McMahan, I.~Mironov, K.~Talwar, and
  L.~Zhang.
\newblock Deep learning with differential privacy.
\newblock {\em arXiv:1607.00133}, 2016.

\bibitem{Blum:2005:SuLQ}
A.~Blum, C.~Dwork, F.~McSherry, and K.~Nissim.
\newblock Practical privacy: The {SuLQ} framework.
\newblock In {\em PODS}, pages 128--138. ACM, 2005.

\bibitem{cho2011friendship}
E.~Cho, S.~A. Myers, and J.~Leskovec.
\newblock Friendship and mobility: user movement in location-based social
  networks.
\newblock In {\em SIGKDD}, pages 1082--1090. ACM, 2011.

\bibitem{cranshaw2010bridging}
J.~Cranshaw, E.~Toch, J.~Hong, A.~Kittur, and N.~Sadeh.
\newblock Bridging the gap between physical location and online social
  networks.
\newblock In {\em UbiComp}. ACM, 2010.

\bibitem{DeMontjoye2013locationunique}
Y.-A. de~Montjoye, C.~A. Hidalgo, M.~Verleysen, and V.~D. Blondel.
\newblock Unique in the crowd: The privacy bounds of human mobility.
\newblock {\em Scientific Reports}, 2013.

\bibitem{dwork2006differential}
C.~Dwork.
\newblock Differential privacy.
\newblock In {\em Automata, languages and programming}, pages 1--12. Springer,
  2006.

\bibitem{dwork2006calibrating}
C.~Dwork, F.~McSherry, K.~Nissim, and A.~Smith.
\newblock Calibrating noise to sensitivity in private data analysis.
\newblock In {\em TCC}, pages 265--284. Springer, 2006.

\bibitem{erlingsson2014rappor}
{\'U}.~Erlingsson, V.~Pihur, and A.~Korolova.
\newblock {RAPPOR}: Randomized aggregatable privacy-preserving ordinal
  response.
\newblock In {\em SIGSAC}, pages 1054--1067. ACM, 2014.

\bibitem{gehrke2012crowd}
J.~Gehrke, M.~Hay, E.~Lui, and R.~Pass.
\newblock Crowd-blending privacy.
\newblock In {\em Advances in Cryptology}, pages 479--496. Springer, 2012.

\bibitem{ghinita2008private}
G.~Ghinita, P.~Kalnis, A.~Khoshgozaran, C.~Shahabi, and K.-L. Tan.
\newblock Private queries in location based services: anonymizers are not
  necessary.
\newblock In {\em SIGMOD}, pages 121--132. ACM, 2008.

\bibitem{gruteser2003anonymous}
M.~Gruteser and D.~Grunwald.
\newblock Anonymous usage of location-based services through spatial and
  temporal cloaking.
\newblock In {\em MobiSys}, pages 31--42. ACM, 2003.

\bibitem{kazemi2012geocrowd}
L.~Kazemi and C.~Shahabi.
\newblock {GeoCrowd}: enabling query answering with spatial crowdsourcing.
\newblock In {\em SIGSPATIAL 2012}, pages 189--198. ACM, 2012.

\bibitem{Korolova:2009:releasingsearch}
A.~Korolova, K.~Kenthapadi, N.~Mishra, and A.~Ntoulas.
\newblock Releasing search queries and clicks privately.
\newblock In {\em WWW}, pages 171--180. ACM, 2009.

\bibitem{leung2010personalized}
K.~W.-T. Leung, D.~L. Lee, and W.-C. Lee.
\newblock Personalized web search with location preferences.
\newblock In {\em ICDE}, pages 701--712. IEEE, 2010.

\bibitem{mokbel2006new}
M.~F. Mokbel, C.-Y. Chow, and W.~G. Aref.
\newblock The new {C}asper: query processing for location services without
  compromising privacy.
\newblock In {\em VLDB}, pages 763--774, 2006.

\bibitem{Nissim:2007:SmoothSensitivity}
K.~Nissim, S.~Raskhodnikova, and A.~Smith.
\newblock Smooth sensitivity and sampling in private data analysis.
\newblock In {\em STOC}, pages 75--84. ACM, 2007.

\bibitem{pham2013ebm}
H.~Pham, C.~Shahabi, and Y.~Liu.
\newblock Inferring social strength from spatiotemporal data.
\newblock {\em ACM Trans. Database Syst.}, 41(1):7:1--7:47, Mar. 2016.

\bibitem{shannon1948mathematical}
C.~E. Shannon and W.~Weaver.
\newblock A mathematical theory of communication, 1948.

\bibitem{sweeney2002k}
L.~Sweeney.
\newblock k-anonymity: A model for protecting privacy.
\newblock {\em International Journal of Uncertainty, Fuzziness and
  Knowledge-Based Systems}, 10(05), 2002.

\bibitem{to2016real}
H.~To, L.~Fan, L.~Tran, and C.~Shahabi.
\newblock Real-time task assignment in hyperlocal spatial crowdsourcing under
  budget constraints.
\newblock In {\em PerCom}. IEEE, 2016.

\bibitem{to2014framework}
H.~To, G.~Ghinita, and C.~Shahabi.
\newblock A framework for protecting worker location privacy in spatial
  crowdsourcing.
\newblock {\em VLDB}, 7(10):919--930, 2014.

\bibitem{to2016dple}
H.~To, K.~Nguyen, and C.~Shahabi.
\newblock Differentially private publication of location entropy.
\newblock {\em University of Southern California, Report ID 16-968, 2016.
  https://www.cs.usc.edu/research/technical-reports-list}.

\bibitem{to2015server}
H.~To, C.~Shahabi, and L.~Kazemi.
\newblock A server-assigned spatial crowdsourcing framework.
\newblock {\em TSAS}, 1(1):2, 2015.

\bibitem{toch2010empirical}
E.~Toch, J.~Cranshaw, P.~H. Drielsma, J.~Y. Tsai, P.~G. Kelley, J.~Springfield,
  L.~Cranor, J.~Hong, and N.~Sadeh.
\newblock Empirical models of privacy in location sharing.
\newblock In {\em UbiComp}, pages 129--138. ACM, 2010.

\bibitem{van2001entropy}
H.~Van Dyke~Parunak and S.~Brueckner.
\newblock Entropy and self-organization in multi-agent systems.
\newblock In {\em AAMAS}, pages 124--130. ACM, 2001.

\bibitem{wang2004entropy}
H.~Wang, K.~Yao, G.~Pottie, and D.~Estrin.
\newblock Entropy-based sensor selection heuristic for target localization.
\newblock In {\em IPSN}, pages 36--45. ACM, 2004.

\bibitem{xiao2015protecting}
Y.~Xiao and L.~Xiong.
\newblock Protecting locations with differential privacy under temporal
  correlations.
\newblock In {\em CCS}, pages 1298--1309. ACM, 2015.

\bibitem{xu2009feeling}
T.~Xu and Y.~Cai.
\newblock Feeling-based location privacy protection for location-based
  services.
\newblock In {\em CCS}, pages 348--357. ACM, 2009.

\bibitem{yanai2009visual}
K.~Yanai, H.~Kawakubo, and B.~Qiu.
\newblock A visual analysis of the relationship between word concepts and
  geographical locations.
\newblock In {\em CIVR}, page~13. ACM, 2009.

\end{thebibliography}
\end{footnotesize}

\balance

\end{document}